\documentclass[12pt, draftclsnofoot, onecolumn]{IEEEtran}
\usepackage{amsfonts,amsmath,amssymb,algorithm,algorithmic,bm}
\usepackage{graphicx,epsfig,cite}
\usepackage{color}
\usepackage{subfigure}
\usepackage[percent]{overpic}
\usepackage{multirow}
\usepackage{autobreak}
\usepackage{hyperref}
\usepackage{amsthm}
\hypersetup{hidelinks}
\newtheorem{theorem}{Theorem}

\newtheorem{prop}{Proposition}
\newtheorem{lemma}{Lemma}
\newtheorem{remark}{Remark}

\def \beqi{\begin{IEEEeqnarray}{rcl}\IEEEyesnumber}
\def \eeqi{\end{IEEEeqnarray}}

\def \bmat{\begin{bmatrix}}
\def \emat{\end{bmatrix}}

\newcommand{\dif}{\mathop{}\!\mathrm{d}}
\DeclareMathOperator{\Tr}{Tr}
\begin{document}

\title{RIS-Enhanced Spectrum Sensing: How Many Reflecting Elements Are Required to Achieve a Detection Probability Close to $1$?}

\author{Jungang Ge, and Ying-Chang Liang,~\IEEEmembership{Fellow,~IEEE}
%\thanks{This research was supported in part by National Natural Science Foundation of China under Grants 61631005 and U1801261.}
\thanks{This work has been submitted to the IEEE for possible publication. Copyright may be transferred without notice, after which this version may no longer be accessible.}
\thanks{J. Ge and Y.-C. Liang are with the Center for Intelligent Networking and Communications (CINC), University of Electronic Science and Technology of China (UESTC), Chengdu 611731, China (e-mail: {gejungang@std.uestc.edu.cn; liangyc@ieee.org}).}
}
\maketitle
\IEEEpubidadjcol
\vspace{-1cm}
\begin{abstract}
In this paper, we propose an reconfigurable intelligent surface (RIS) enhanced spectrum sensing system, in which the primary transmitter is equipped with single antenna, the secondary transmitter is equipped with multiple antennas, and the RIS is employed to improve the detection performance. Without loss of generality, we adopt the maximum eigenvalue detection approach, and propose a corresponding analytical framework based on large dimensional random matrix theory, to evaluate the detection probability in the asymptotic regime. Besides, the phase shift matrix of the RIS is designed with only the statistical channel state information (CSI), which is shown to be quite effective when the RIS-related channels are of Rician fading or line-of-sight (LoS). With the designed phase shift matrix, the asymptotic distributions of the equivalent channel gains are derived. Then, we provide the theoretical predictions about the number of reflecting elements (REs) required to achieve a detection probability close to $1$. Finally, we present the Monte-Carlo simulation results to evaluate the accuracy of the proposed asymptotic analytical framework for the detection probability and the validity of the theoretical predictions about the number of REs required to achieve a detection probability close to $1$. Moreover, the simulation results show that the proposed RIS-enhanced spectrum sensing system can substantially improve the detection performance.
\end{abstract}

\begin{IEEEkeywords}
Reconfigurable intelligent surface, spectrum sensing, statistical phase shift design.
\end{IEEEkeywords}

\section{Introduction} \label{sec:intro}
% spectrum sensing

% RIS, passive beamforming statistical design

% why RIS for spectrum sensing,
% 1. enhance the sensing SNR, improve the sensing performance with the same amount of signal samples
% 2. reduce the number of signal samples to achieve a high detection probability, leave more time resource for secondary data transmission

The ever-increasing spectrum demand to support ubiquitous communications will last in current and future wireless communication systems. Cognitive radio (CR), which is regarded as a key technology enabling efficient spectrum management \cite{liang2020dynamic}, can substantially improve the spectrum efficiency \cite{liang2008sensing}. In opportunistic CR, the secondary users are allowed to reuse the licensed spectrum allocated to the primary users when the primary users are inactive. Hence, the spectrum sensing technique, which helps the secondary users detect the activities of the primary users, becomes a fundamental technology of CR. As a consequence, spectrum sensing has attracted a lot of attention from both academia and industry in the past two decades.

To realize spectrum sensing of high quality under various practical constraints, many spectrum sensing algorithms have been proposed, such as energy detection \cite{sonnenschein1992radiometric, sahai2005maximum}, matched filtering (MF)-based methods \cite{sahai2005maximum, chen2007signature}, cyclostationary detection \cite{gardner1991exploitation, han2006spectral}, eigenvalue-based detection methods \cite{zeng2008maximum, zeng2009eigenvalue, zhang2010multi, bouallegue2017blind, yucek2009survey, zeng2010review, awin2019blind}. The detection performance can be improved by designing complex test statistics, and a recent line of works try to find optimal test statistics by exploiting the deep learning technique \cite{liu2019deep, gao2019deep}.
%Besides, the corresponding detection thresholds are usually determined by the given $P_{fa}$ and the distributions of the test statistics under where the primary user are inactive.
On the other hand, there are also many results for analyzing the detection performance of some spectrum sensing algorithms \cite{bianchi2011performance, jin2012performance, wei2012spectrum, wei2014multi, sedighi2015performance, sedighi2016eigenvalue}. In particular, \cite{bianchi2011performance} provides the analytical results for the error exponent of the generalized likelihood ratio test (GLRT) with the knowledge of large dimensional random matrix theory. The mathematical expressions of the false alarm probability and the detection probability are derived for the covariance-based detector \cite{jin2012performance}, the Hadamard ratio detector \cite{sedighi2015performance}, and Wilks' detector \cite{wei2014multi}.
%\cite{jin2012performance} gives the mathematical expressions of the detection probability and the false alarm probability of the covariance based detection. Similarly, \cite{sedighi2015performance} derives the analytical detection and false alarm probabilities of the Hadamard ratio detector. In \cite{wei2012spectrum}, the false alarm and detection probabilities of the spherical test are obtained in closed form. \cite{wei2014multi} derives the accurate approximations of the false alarm and detection probabilities of Wilks' detector under where the noise variance is arbitrary.
Moreover, \cite{sedighi2016eigenvalue} proposes two eigenvalue-based detectors with invariant constant false alarm rate, using higher order moments of the eigenvalues of the sample covariance matrix. Besides, the closed-form formulas of the false alarm and detection probabilities of these two detectors are obtained with moment-based approximations. From the aforementioned results, it can be observed that the number of signal samples required to achieve a high detection probability is much larger than the number of sensing antennas when the sensing signal-to-noise ratio (SNR) is extremely low. Taking \cite{zeng2008maximum} as an example, to achieve a detection probability close to $1$ at a SNR as low as $-18dB$, the number of signal samples should be more than $1.6\times10^{5}$ when the number of sensing antennas is $8$ (see Fig. $5$ in \cite{zeng2008maximum}). In addition, as pointed out in \cite{liang2008sensing}, using more signal samples for sensing means better detection performance but less time resource for the secondary data transmission. Therefore, how to reduce the number of signal samples required to achieve a high detection probability is quite significant to improve the data transmission rate of the secondary users.

In conventional spectrum sensing approaches, many complex test statistics are designed to improve the detection performance since the sensing algorithms have to be adapted to the uncontrollable wireless environment. Recently, the emerged reconfigurable intelligent surface (RIS) provides us an efficient way to make the wireless environment programmable \cite{liang2019large, wu2019towards, gong2020toward}. Besides, the recent advances in micro-materials show that the phase shifts of RIS can adapt to the changes of the wireless environment in real-time, thereby realize the reconfiguration of the wireless environment \cite{cui2014coding}. With the instantaneous channel state information (CSI), the RIS can be employed in various communication scenarios \cite{chen2019intelligent, guo2020weighted, huang2019reconfigurable, yuan2020intelligent}. In particular, \cite{chen2019intelligent} proposes to exploit the RIS to realize physical layer security via solving a minimum-secrecy-rate maximization problem. \cite{huang2019reconfigurable} considers a resource allocation problem to maximize the energy efficiency in an RIS-assisted multiple-input multiple-output (MIMO) downlink communication system. In \cite{guo2020weighted}, it is shown that the RIS can be utilized to enhance the weighted sum-rate of the multiuser multiple-input single-output (MISO) downlink communication system. Besides, \cite{yuan2020intelligent} proposes an RIS-assisted cognitive radio system, in which RIS is employed to improve the achievable rates of the secondary users. However, it is quite hard to acquire accurate instantaneous CSI of the RIS-related channels separately in practice because of the difficulty of the channel estimation \cite{chen2019channel}. Therefore, many robust approaches are proposed to deal with the imperfect CSI due to the channel estimation errors \cite{guo2020weighted, yuan2020intelligent, zhou2020framework, zhou2020robust}. Besides, the RIS is usually installed on the exterior walls of buildings, high towers, etc., to effectively reflect the electromagnetic (EM) wave in the air. Consequently, the RIS-related channels are likely to have constant line-of-sight (LoS) components. Hence, there are also many works studying how to design the RIS-assisted communication systems with only the statistical CSI \cite{kammoun2020asymptotic, zhang2020transmitter, zhang2021large}.

% contribution
% 1. RIS-aided spectrum sensing system
% 2. propose a theoretical framework to evaluate $P_d$ in the asymptotic regime
% 3. propose a statistical design for the RIS phase shift matrix under LoS and Rician channels and the asymptotic distributions of corresponding channel gains
% 4. give an analytical results on the number of REs to achieve near-perfect detection

In this paper, we propose an RIS-enhanced spectrum sensing system, in which the primary transmitter (PT) is equipped with signal antenna and the secondary transmitter (ST) performs spectrum sensing with multiple antennas. The RIS with multiple REs, whose phase shifts are designed with only the statistical CSI, is employed to improve the detection performance by enhancing the average sensing SNR. Our main contributions are summarized as follows.

\begin{itemize}
  \item To the best of our knowledge, the proposed RIS-enhanced spectrum sensing system is one of the first works \cite{li2020irs} that employ RIS in the spectrum sensing systems. This work is particularly applicable to IEEE 802.22 scenarios in which primary TV tower and spectrum sensor are static \cite{chen2007signature}. By introducing the RIS with multiple REs into the conventional spectrum sensing systems, the detection probability can be substantially improved with the same number of signal samples.
  \item Without loss of generality, we adopt the maximum eigenvalue detection (MED) approach \cite{zeng2008maximum} to demonstrate the positive impact of the RIS on the detection performance. Correspondingly, with the knowledge of the spiked model from random matrix theory, we intuitively propose an analytical framework for evaluating the detection probability in the asymptotic regime, where both the number of antennas at ST and that of the signal samples go to infinity with a limit ratio.
  \item For the RIS-enhanced spectrum sensing system, we propose to design the phase shift matrix of the RIS with only the statistical CSI. We will see that the statistical design scheme is quite effective when the RIS-related channels are of Rician fading or LoS. In addition, we derive the asymptotic distributions of the equivalent channel gain under three possible conditions of the RIS-related channels, i.e., LoS, Rayleigh fading, and Rician fading.
  \item According to the proposed analytical framework for evaluating the detection probability and the derived asymptotic distributions of the equivalent channel gains, we provide the theoretical predictions about the number of REs required to achieve a detection probability close to $1$. The Monte-Carlo simulation results justify the validity of the theoretical predictions.
\end{itemize}

The remainder of this paper is organized as follows. Section \ref{sec:RIS4SS} illustrates the RIS-enhanced spectrum sensing system, and presents details of the maximum eigenvalue detection, including the proposed asymptotic analytical framework for evaluating the detection probability. Section \ref{sec:CGA} introduces the statistical phase shift design, and derives the asymptotic distributions of the equivalent channel gains under different RIS-related channel conditions. Section \ref{sec:perfectdetection} analyzes the necessary condition and the sufficient condition to achieve a detection probability close to $1$, and provides the theoretical predictions about the number of REs required for meeting these two conditions, respectively. In Section \ref{sec:sim}, we provide the Monte-Carlo simulation results to evaluate the validity of the analytical results. Besides, we investigate the impact of the number of signal samples and the channel characteristics on the number of REs required to achieve a high detection probability. Finally, Section \ref{sec:conclu} concludes this paper.

\textbf{Notations:} The notations in this paper are as follows. Bold uppercase and lowercase symbols (e.g., $\mathbf{x}$ and $\mathbf{X}$) are used to represent the column vectors and matrices, respectively. $\mathbf{X}^H$ denotes the Hermitian of $\mathbf{X}$, and $x^{\dagger}$ denotes the conjugate of $x$. In addition, $\|\cdot\|$, $\mathbb{E}[\cdot]$, $\Tr(\cdot)$ denote the Euclidean norm operator, the expectation operator and the trace operator, respectively. Further,
$\mathbf{I}_{M}$ denotes an $M$-dimensional identity matrix. $\mathbf{x}\sim\mathcal{CN}(\mathbf{0}, \mathbf{\Sigma})$ means that $\mathbf{x}$ is a complex Gaussian random vector with zero mean and covariance matrix $\mathbf{\Sigma}$. ${\rm diag}([x_1, \cdots, x_M])$ denotes the diagonal matrix with $x_1, \cdots, x_M$ the diagonal elements.

\section{RIS-enhanced Spectrum Sensing}\label{sec:RIS4SS}

\subsection{Signal Model} \label{subsec:sigmodel}

\begin{figure}[!t]
\begin{center}
\epsfxsize=0.25\textwidth \leavevmode
\epsffile{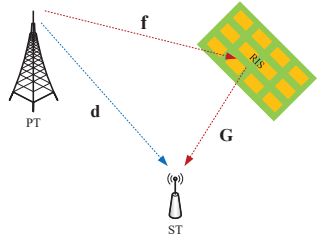}
\caption{An RIS-enhanced spectrum sensing system.}\label{fig:sysmodel}
\end{center}
\end{figure}

As shown in Fig. \ref{fig:sysmodel}, we consider a RIS-enhanced spectrum sensing system, where the PT is equipped with single antenna and the ST is equipped with $N$ antennas. Besides, the RIS has $M$ passive REs to reflect the incident EM wave. Here, we use $\mathbf{d}\in\mathbb{C}^{N}$, $\mathbf{f}\in\mathbb{C}^{M}$ and $\mathbf{G}\in\mathbb{C}^{N\times M}$ to denote the channels from PT to ST, from PT to RIS, and from RIS to ST, respectively. During the sensing interval, the ST can collect $n$ signal samples to determine whether the PT is active. Obviously, there are two hypotheses: $\mathcal{H}_0$, PT is inactive; and $\mathcal{H}_1$, PT is active. Thus, the received signal samples at the ST under $\mathcal{H}_0$ and $\mathcal{H}_1$ can be respectively written as
\begin{align}
  &\mathcal{H}_0: \mathbf{x}_k=\mathbf{u}_k, \\
  &\mathcal{H}_1: \mathbf{x}_k=(\mathbf{d}+\mathbf{G}\mathbf{\Phi}\mathbf{f})s_k + \mathbf{u}_k,
\end{align}
where $\mathbf{x}_k$ denotes the $k$-th sample, $s_k$ is the $k$-th transmitted symbol of PT with $s_k\sim\mathcal{CN}(0,\sigma_s^2)$; $\mathbf{u}_k\sim\mathcal{CN}(\mathbf{0}, \sigma_u^2\mathbf{I}_N)$ denotes the additive white Gaussian noise (AWGN) at ST; $\mathbf{\Phi}={\rm diag}([e^{j\phi_1},$ $\cdots, e^{j\phi_M}])$ denotes the phase shift matrix of the RIS, whose elements are considered to have unit-modulus \cite{huang2019reconfigurable}.

\subsection{Maximum Eigenvalue Detection}\label{subsec:MED}
In this paper, we consider the maximum eigenvalue detection (MED) \cite{zeng2008maximum}, to evaluate the impact of RIS on the detection performance. The MED method originates from the difference of the {\it population covariance matrices} under the two hypotheses, which are defined as
\begin{align}
&\mathcal{H}_0:\mathbf{R}_{\mathbf{x}\mathbf{x}}^{0}=\mathbb{E}[\mathbf{x}_k\mathbf{x}_k^H]=\mathbb{E}[\mathbf{u}_k\mathbf{u}_k^H]=\sigma_u^2\mathbf{I}_N, \label{eq:PCMH0}\\
&\mathcal{H}_1:\mathbf{R}_{\mathbf{x}\mathbf{x}}^{1}=\mathbb{E}[\mathbf{x}_k\mathbf{x}_k^H]
=(\mathbf{d}+\mathbf{G}\mathbf{\Phi}\mathbf{f})(\mathbf{d}+\mathbf{G}\mathbf{\Phi}\mathbf{f})^H\sigma_s^2 +\sigma_u^2\mathbf{I}_N.\label{eq:PCMH1}
\end{align}
Obviously, $(\mathbf{d}+\mathbf{G}\mathbf{\Phi}\mathbf{f})(\mathbf{d}+\mathbf{G}\mathbf{\Phi}\mathbf{f})^H$ is a rank-$1$ matrix, $\mathbf{R}_{\mathbf{x}\mathbf{x}}^{1}$ has only one eigenvalue that not equals to $\sigma_u^2$, namely, the largest eigenvalue of $\mathbf{R}_{\mathbf{x}\mathbf{x}}^{1}$, which equals $\|\mathbf{d}+\mathbf{G}\mathbf{\Phi}\mathbf{f}\|^2\sigma_s^2+\sigma_u^2$. Hence, the active PT can be detected by the largest eigenvalue of the population covariance matrix.
In practice, we only have access to the {\it sample covariance matrix}, i.e.,
\begin{equation}\label{eq:SCM}
\hat{\mathbf{R}}_{\mathbf{x}\mathbf{x}}=\frac{1}{n}\sum_{k=1}^{n}\mathbf{x}_k\mathbf{x}_k^H.
\end{equation}
When $N$ is fixed and $n$ goes to infinity, $\hat{\mathbf{R}}_{\mathbf{x}\mathbf{x}}$ is a quite accurate approximation of $\mathbf{R}_{\mathbf{x}\mathbf{x}}$, therefore the eigenvalues of $\hat{\mathbf{R}}_{\mathbf{x}\mathbf{x}}$ approximate that of $\mathbf{R}_{\mathbf{x}\mathbf{x}}$ well. However, when $n$ is not sufficiently large, $\hat{\mathbf{R}}_{\mathbf{x}\mathbf{x}}$ deviates significantly from $\mathbf{R}_{\mathbf{x}\mathbf{x}}$. The eigenvalues of $\hat{\mathbf{R}}_{\mathbf{x}\mathbf{x}}$ are actually random variables, which can be characterized with the results from random matrix theory \cite{ge2021large}.

%Spectrum sensing is actually a conventional signal detection problem, in which

In spectrum sensing, we mainly concern about two probabilities, namely, the false alarm probability $P_{fa}$ and the detection probability $P_{d}$. Denoting the two possible detection results, i.e, PT is inactive and PT is active, by $\mathcal{D}_0$ and $\mathcal{D}_1$, the two probabilities are respectively defined as
\begin{equation}\label{eq:Pfa}
P_{fa} = P(\mathcal{D}_1|\mathcal{H}_0),\ P_{d} = P(\mathcal{D}_1|\mathcal{H}_1).
\end{equation}
%\begin{equation}\label{eq:Pd}
%P_{d} = P(\mathcal{D}_1|\mathcal{H}_1).
%\end{equation}
Denoting $\lambda_{\max}$ the largest eigenvalue of the sample covariance matrices, the test statistic of MED is defined as
\begin{equation}\label{eq:teststatistic}
T\triangleq \frac{\lambda_{\max}}{\sigma_u^2},
\end{equation}
where $\sigma_u^2$ is used to normalize the largest eigenvalue for simplicity's sake.

\subsection{Detection Threshold Determination}\label{subsec:detectionthreshold}
The test statistic, $T$, can be regarded as the largest eigenvalue of $\hat{\mathbf{R}}_{\mathbf{x}\mathbf{x}}$ when $\sigma_u^2$ is set to $1$.
According to Neyman-Pearson lemma, the detection threshold $\gamma$ can be determined with a given constraint on $P_{fa}$, say, $P_{fa}\leq \alpha$. Specifically, $\gamma$ can be calculated by solving
\begin{equation}\label{eq:calgamma}
P(T>\gamma|\mathcal{H}_0)=\alpha.
\end{equation}
Noting that the sample covariance matrix under $\mathcal{H}_0$ is actually a Wishart matrix \cite{ge2021large}, the asymptotic distribution of the largest eigenvalue is given as follows.

\begin{prop}\label{prop:pdflargestWishart}
Let $\mathbf{X}_N\in\mathbb{C}^{N\times n}$ be a random matrix whose entries are independent and identically distributed ({\it i.i.d.}) Gaussian random variables with zero mean and variance $1/n$. Denoting the largest eigenvalue of the Wishart matrix $\mathbf{X}_N\mathbf{X}_N^H$ by $\lambda_{N}^{+}$, as $N,n\to\infty$ with $c=\lim N/n<1$, we have
\begin{equation}\label{eq:pdflargestWishart}
N^{\frac{2}{3}}\frac{\lambda_{N}^{+}-(1+\sqrt{c})^2}{(1+\sqrt{c})^{\frac{4}{3}}\sqrt{c}}\xrightarrow[N,n\to\infty]{\mathcal{D}} X\sim F_2,
\end{equation}
where $\xrightarrow[N,n\to\infty]{\mathcal{D}}$ means convergence in distribution as $N,n\to\infty$ and $F_2$ is the {\it Tracy-Widom} distribution of order $2$ \cite{ge2021large}.
\end{prop}

With Proposition \ref{prop:pdflargestWishart}, the detection threshold can be determined by
\begin{equation}\label{eq:gamma}
\gamma = N^{-\frac{2}{3}}(1+\sqrt{c})^{\frac{4}{3}}\sqrt{c}F_{2}^{-1}\left(1-\alpha\right)+(1+\sqrt{c})^2,
\end{equation}
where the quantile function, i.e., $F_{2}^{-1}(\cdot)$, can be calculated with the tools provided by \cite{RMTstat}.

\subsection{Detection Probability Evaluation}\label{subsec:detectionproba}

To evaluate the detection probability, we can resort to the {\it single spiked model} in random matrix theory. The single spiked model can be briefly described as follows \cite{ge2021large}.
\begin{prop}\label{prop:singlespikedmodel}
Let $\mathbf{T}_N$ be a fixed $N\times N$ non-negative definite Hermitian matrix, $\mathbf{X}_{N}\in\mathbb{C}^{N\times n}$ be a random matrix whose entries $\mathbf{X}_{N, ij}$ are {\it i.i.d.} complex random variables such that
\begin{equation*}
\mathbb{E}(\mathbf{X}_{N,11})=0,\ \mathbb{E}(|\mathbf{X}_{N,11}|^2)=1,\ {\rm and}\ \mathbb{E}(|\mathbf{X}_{N,11}|^4)<\infty.
\end{equation*}
Let $\mathbf{B}_N=\frac{1}{n}\mathbf{T}_N^{\frac{1}{2}}\mathbf{X}_N\mathbf{X}_N^H\mathbf{T}_N^{\frac{1}{2}}$ denote the sample covariance matrix where $\mathbf{T}_N^{\frac{1}{2}}$ is a Hermitian square root of $\mathbf{T}_N$, with eigen-decomposition, we have $\mathbf{S}\triangleq\mathbf{U}_{\mathbf{B}}\mathbf{B}_N\mathbf{U}_{\mathbf{B}}^{-1}={\rm diag}(s_1^{(N)},s_2^{(N)},\cdots, s_N^{(N)})$.
%\begin{equation*}
%\mathbf{S}\triangleq\mathbf{U}_{\mathbf{B}}\mathbf{B}_N\mathbf{U}_{\mathbf{B}}^{-1}={\rm diag}(s_1^{(N)},s_2^{(N)},\cdots, s_N^{(N)}).
%\end{equation*}
For definiteness, we order the eigenvalues as $s_1^{(N)}\geq s_2^{(N)}\geq \cdots\geq s_N^{(N)}$.
Similarly, for some unitary matrix $\mathbf{U}_{\mathbf{T}}$, we have $\mathbf{\Sigma}\triangleq\mathbf{U}_{\mathbf{T}}\mathbf{T}_N\mathbf{U}_{\mathbf{T}}^{-1}={\rm diag}(\tau, 1, \cdots, 1).$
%\begin{equation*}
%\mathbf{\Sigma}\triangleq\mathbf{U}_{\mathbf{T}}\mathbf{T}_N\mathbf{U}_{\mathbf{T}}^{-1}={\rm diag}(\tau, 1, \cdots, 1).
%\end{equation*}
With $N, n\to\infty$ and $\lim N/n\to c$,
\begin{itemize}
  \item when $\tau<\sqrt{c}+1$,
  $$N^{\frac{2}{3}}\frac{s_1^{(N)}-(1+\sqrt{c})^2}{(1+\sqrt{c})^{\frac{4}{3}}\sqrt{c}}\xrightarrow[N,n\to\infty]{\mathcal{D}}X\sim F_2,$$
  where $F_2$ denotes the Tracy-Widom distribution of order $2$;
  \item when $\tau>\sqrt{c}+1$,
  $$s_1^{(N)}\xrightarrow[N,n\to\infty]{\mathcal{D}}\mathcal{N}(\mu_s, v_s),$$
  where $$\mu_s = \tau+\frac{c\tau}{\tau-1},\  v_s=\frac{\tau^2}{n}\left(1-\frac{c}{(\tau-1)^2}\right).$$
\end{itemize}
\end{prop}
According to Proposition \ref{prop:singlespikedmodel}, the sample covariance matrix under $\mathcal{H}_1$ can be modeled with the single spiked model. More specifically, when $\sigma_u^2=1$, the distribution of $T$ under $\mathcal{H}_1$ can be characterized with the single spiked model with $\tau=\|\mathbf{d}+\mathbf{G}\mathbf{\Phi}\mathbf{f}\|^2\sigma_s^2+1$. Further, let $\sigma_s^2=1$ and $g\triangleq\|\mathbf{d}+\mathbf{G}\mathbf{\Phi}\mathbf{f}\|^2$, we have the following theorem.

\begin{theorem}\label{thm:spike}
Consider a single spiked model where $\tau = g+1$, we have the following conclusions.
\begin{itemize}
  \item When $g+1<\sqrt{c}+1$, i.e., $g<\sqrt{c}$,
  \begin{equation}\label{eq:gdfleq}
    N^{\frac{2}{3}}\frac{T-(1+\sqrt{c})^2}{(1+\sqrt{c})^{\frac{4}{3}}\sqrt{c}}\xrightarrow[N,n\to\infty]{\mathcal{D}}X\sim F_2,
  \end{equation}
  where $F_2$ denotes the Tracy-Widom distribution of order $2$.
  \item When $g+1>\sqrt{c}+1$, i.e., $g>\sqrt{c}$,
  \begin{equation}\label{eq:gdfgeq}
    T\xrightarrow[N,n\to\infty]{\mathcal{D}}X\sim\mathcal{N}\left(\mu_T(g), v_T(g)\right),
  \end{equation}
  where
  \begin{equation}\label{eq:mTspike}
    \mu_T(g) = g+1+c+\frac{c}{g},\ v_T=\frac{(g+1)^2}{n}\left(1-\frac{c}{g^2}\right).
  \end{equation}
%  \begin{equation}\label{eq:vTspike}
%    v_T=\frac{(g+1)^2}{n}\left(1-\frac{c}{g^2}\right).
%  \end{equation}
\end{itemize}
\end{theorem}
\begin{remark}\label{rmk:spikedetection}
%An important observation from Theorem \ref{thm:spike} is that, i
If the equivalent channel gain is very weak such that $g<\sqrt{c}$, the asymptotic distributions of the largest sample eigenvalue under $\mathcal{H}_0$ and $\mathcal{H}_1$ are the same, the MED method will fail to detect the activity of PT. As a consequence, we get a detection probability that is equal to the false alarm probability, i.e., $P_{fa}=P_d$. This actually provides us two effective ways to improve the performance of sensing algorithms. On the one hand, we can reduce $c$ by increasing the number of signal samples to let $g>\sqrt{c}$. This is widely observed in the existing literatures about spectrum sensing. On the other hand, we can increase the channel gain to achieve the same goal, and this is exactly the motivation of RIS-enhanced spectrum sensing systems.
\end{remark}

Theorem \ref{thm:spike} tells us, when $g$ is not large enough to exceed $\sqrt{c}$, the largest eigenvalue behaves the same as under $\mathcal{H}_0$. When $g>\sqrt{c}$, the largest sample eigenvalue satisfies a spiked distribution as shown in \eqref{eq:gdfgeq}. Therefore, the detection probability is given by
\begin{align}\label{eq:Pdint}
P_d &=\int_{\gamma}^{\infty}p(T|\mathcal{H}_1)\dif T =\int_{\gamma}^{\infty}\int p(T|g)p(g)\dif g\dif T\nonumber\\
&=\int_{\gamma}^{\infty}\int_{-\infty}^{\sqrt{c}} p(T|g)p(g)\dif g\dif T+\int_{\gamma}^{\infty}\int_{\sqrt{c}}^{\infty} p(T|g)p(g)\dif g\dif T\nonumber\\
&= P_{fa}\int_{-\infty}^{\sqrt{c}} p(g)\dif g + \int_{\gamma}^{\infty}\int_{\sqrt{c}}^{\infty} p(T|g)p(g)\dif g\dif T.
\end{align}

It is worth noting that Theorem \ref{thm:spike} and \eqref{eq:Pdint} are derived in the asymptotic regime, where $N, n\to\infty$ with $\lim N/n=c$. Obviously, the detection probability depends on the statistical characteristic of the equivalent channel gain, i.e., $g$. To evaluate the impact of RIS on the detection performance, we have to clearly specifies the design for the phase shift matrix of the RIS and then analyze the impact of the RIS on the equivalent channel gain.

\section{Statistical Phase Shift Design and Channel Gain Analysis}\label{sec:CGA}

The channel from PT to ST is usually of Rayleigh fading, i.e., $\mathbf{d}\sim\mathcal{CN}(\mathbf{0}, \beta_d\mathbf{I}_N)$ with $\beta_d$ the pathloss from PT to ST. Without RIS, the equivalent channel gain becomes $\|\mathbf{d}\|^2$. For large $N$, with the central limit theorem (CLT), we have \begin{equation}\label{eq:ddfRay}
\|\mathbf{d}\|^2\sim\mathcal{N}\left(\mu_d, v_d\right),
\end{equation}
where $\mu_d=N\beta_d,\ v_d=N\beta_d^2$.
%\begin{equation}\label{eq:md}
%\mu_d=N\beta_d,\ v_d=N\beta_d^2.
%\end{equation}
In the RIS-assisted systems, it is well-known that adaptively tuning the phase shift matrix of the RIS according to the channel variations in real-time can maximize the equivalent channel gain \cite{liang2019large, wu2019towards, gong2020toward}. However, it is impossible for the secondary users to acquire the instantaneous CSI of the primary users in CR systems.
%\begin{equation}\label{eq:ddfRay}
%\|\mathbf{d}\|^2\sim\mathcal{N}\left(\mu_d, v_d\right),
%\end{equation}
%where
%\begin{equation}\label{eq:md}
%\mu_d=N\beta_d,\ v_d=N\beta_d^2.
%\end{equation}
Besides, it is quite difficult to perform separate channel estimation of the RIS-related channels for the low-cost passive REs. Thus, we consider to design the phase shift matrix of the RIS with only the statistical CSI. In addition, to achieve $g>\sqrt{c}$, the phase shift matrix can be designed to maximize the expectation of the equivalent channel gain. Obviously, the statistical phase shift design is most effective when the RIS-related channels, namely, $\mathbf{f}$ and $\mathbf{G}$, are LoS channels. Hence, we start the analysis from the LoS case, where $\mathbf{f}$, $\mathbf{G}$ are LoS channels.

%Moreover, the channel from PT to ST is usually of Rayleigh fading, i.e., $\mathbf{d}\sim\mathcal{CN}(\mathbf{0}, \beta_d\mathbf{I}_N)$ with $\beta_d$ the pathloss from PT to ST. Without RIS, the equivalent channel gain becomes $\|\mathbf{d}\|^2$. For large $N$, with the central limit theorem (CLT), we have $\|\mathbf{d}\|^2\sim\mathcal{N}\left(\mu_d, v_d\right)$ with $\mu_d=N\beta_d,\ v_d=N\beta_d^2$.

\subsection{Under Line-of-Sight (LoS) RIS-related Channels}\label{subsec:LoS}

When $\mathbf{f}$, $\mathbf{G}$ are LoS channels, we have $\mathbf{f}=\sqrt{\beta_f}\mathbf{a}_{f}$ with $\beta_f$ the pathloss from PT to RIS,
%\begin{equation}\label{eq:LoSf}
%\mathbf{f}=\sqrt{\beta_f}\mathbf{a}_{f},
%\end{equation}
where
%$\mathbf{a}_{f}=[1, e^{j2\pi\frac{d}{\lambda}\sin\theta_f^{(AOA)}}, \cdots, e^{j2\pi\frac{d}{\lambda}(M-1)\sin\theta_f^{(AOA)}}]^T$
\begin{equation}\label{eq:LoSfsv}
\mathbf{a}_{f}=[1, e^{j2\pi\frac{d}{\lambda}\sin\theta_f^{(AOA)}}, \cdots, e^{j2\pi\frac{d}{\lambda}(M-1)\sin\theta_f^{(AOA)}}]^T
\end{equation}
denotes the steering vector at RIS and $\theta_f^{(AOA)}$ is the angle of arrival (AOA) of $\mathbf{f}$. Similarly, $\mathbf{G}$ can be decomposed as $\mathbf{G}=\sqrt{\beta_G}\mathbf{a}_{G}\mathbf{b}_{G}^H$ with $\beta_G$ the pathloss from RIS to ST,
%\begin{equation}\label{eq:LoSG}
%\mathbf{G}=\sqrt{\beta_G}\mathbf{a}_{G}\mathbf{b}_{G}^H,
%\end{equation}
where
%$\mathbf{a}_{G}=[1, e^{j2\pi\frac{d}{\lambda}\sin\theta_G^{(AOA)}}, \cdots, e^{j2\pi\frac{d}{\lambda}(N-1)\sin\theta_G^{(AOA)}}]^T$,
%$\mathbf{b}_{G}=[1,e^{j2\pi\frac{d}{\lambda}\sin\theta_G^{(AOD)}}, \cdots, e^{j2\pi\frac{d}{\lambda}(M-1)\sin\theta_G^{(AOD)}}]^T$
\begin{equation}\label{eq:LoSGsva}
\mathbf{a}_{G}=[1, e^{j2\pi\frac{d}{\lambda}\sin\theta_G^{(AOA)}}, \cdots, e^{j2\pi\frac{d}{\lambda}(N-1)\sin\theta_G^{(AOA)}}]^T,
\end{equation}
\begin{equation}\label{eq:LoSGsvb}
\mathbf{b}_{G}=[1,e^{j2\pi\frac{d}{\lambda}\sin\theta_G^{(AOD)}}, \cdots, e^{j2\pi\frac{d}{\lambda}(M-1)\sin\theta_G^{(AOD)}}]^T
\end{equation}
denote the two steering vectors at ST and RIS, respectively; $\theta_G^{(AOA)}$, $\theta_G^{(AOD)}$ are the angle of arrival (AOA) and the angle of departure (AOD) of $\mathbf{G}$, respectively.
Then, $\mathbf{\Phi}$ can be statistically determined by solving a simple optimization problem, i.e.,
\begin{equation}\label{eq:optimalPhiLoS}
\mathbf{\Phi} = \arg\max_{\mathbf{\Phi}}\mathbb{E}_{\mathbf{d}, \mathbf{f}, \mathbf{G}}\left\{\|\mathbf{d}+\mathbf{G}\mathbf{\Phi}\mathbf{f}\|^2\right\}.
\end{equation}
Further, $\mathbf{d}+\mathbf{G}\mathbf{\Phi}\mathbf{f}$ can be rewritten as $\mathbf{d}+\sqrt{\beta_f\beta_G}\mathbf{a}_{G}\mathbf{b}_{G}^H\mathbf{\Phi}\mathbf{a}_{f}$,
\eqref{eq:optimalPhiLoS} reduces to
\begin{align}\label{eq:optimalPhiLoSe}
\mathbf{\Phi} &= \arg\max_{\mathbf{\Phi}}\mathbb{E}_{\mathbf{d}, \mathbf{f}, \mathbf{G}}\left\{\|\mathbf{d}\|^2 + \beta_f\beta_G\|\mathbf{a}_{G}\mathbf{b}_{G}^H\mathbf{\Phi}\mathbf{a}_{f}\|^2\right\}=\arg\max_{\mathbf{\Phi}}\beta_f\beta_G\|\mathbf{a}_{G}\mathbf{b}_{G}^H\mathbf{\Phi}\mathbf{a}_{f}\|^2.
\end{align}
Noting that $\mathbf{b}_{G}^H\mathbf{\Phi}\mathbf{a}_{f}$ is actually a scalar, therefore maximizing $|\mathbf{b}_{G}^H\mathbf{\Phi}\mathbf{a}_{f}|$ is equivalent to \eqref{eq:optimalPhiLoS}. Finally, we obtain the optimal phase shift matrix of RIS as
\begin{equation}\label{eq:optimalPSMLoS}
\phi_i=\exp[-j\angle b_G(i)a_f(i)],\ i=1,2,\cdots, M,
\end{equation}
where $b_G(i)$, $a_f(i)$ denote the $i$-th element of $\mathbf{b}_{G}$ and $\mathbf{a}_{f}$, respectively, and $\angle b_G(i)a_f(i)$ denotes the phase of $b_G(i)a_f(i)$.

With the optimized phase shift matrix, we have $\mathbf{d}+\mathbf{G}\mathbf{\Phi}\mathbf{f}=\mathbf{d}+p\mathbf{a}_{f}$, where $p\triangleq M\sqrt{\beta_f\beta_G}$. Obviously, the cascaded channel, namely, $\mathbf{G}\mathbf{\Phi}\mathbf{f}$, introduces a non-zero constant component $p\mathbf{a}_{f}$, which improves the equivalent channel gain. Thus, the equivalent channel gain satisfies a non-central chi-square distribution. When $M$, $N$ are sufficiently large, we can obtain the following limit distribution of $g$ via the CLT:
\begin{equation}\label{eq:gdfLoS}
g\sim\mathcal{N}(\mu_g^{\rm LoS}, v_g^{\rm LoS}),
\end{equation}
where
\begin{equation}\label{eq:mgLoS}
\mu_g^{\rm LoS} = N\beta_d + M^2N\beta_f\beta_G, \ v_g^{\rm LoS} = 2M^2N\beta_d\beta_f\beta_G.
\end{equation}
%\begin{equation}\label{eq:vgLoS}
%v_g^{\rm LoS} = 2M^2N\beta_d\beta_f\beta_G.
%\end{equation}

To validate the accuracy of \eqref{eq:ddfRay} and \eqref{eq:gdfLoS}, we provide the analytical results and the Monte-Carlo simulation results with $N\beta_d=-20dB, N=64, M=40, \beta_f\beta_G=-60dB$, as shown in Fig. \ref{fig:pdfvad_LoS}. Obviously, the analytical results fit the simulation results well.
\begin{figure}[!t]
\begin{center}
\subfigure[{\it p.d.f.} of $\|\mathbf{d}\|^2$]{%
\epsfxsize=0.2\textwidth \leavevmode
\epsffile{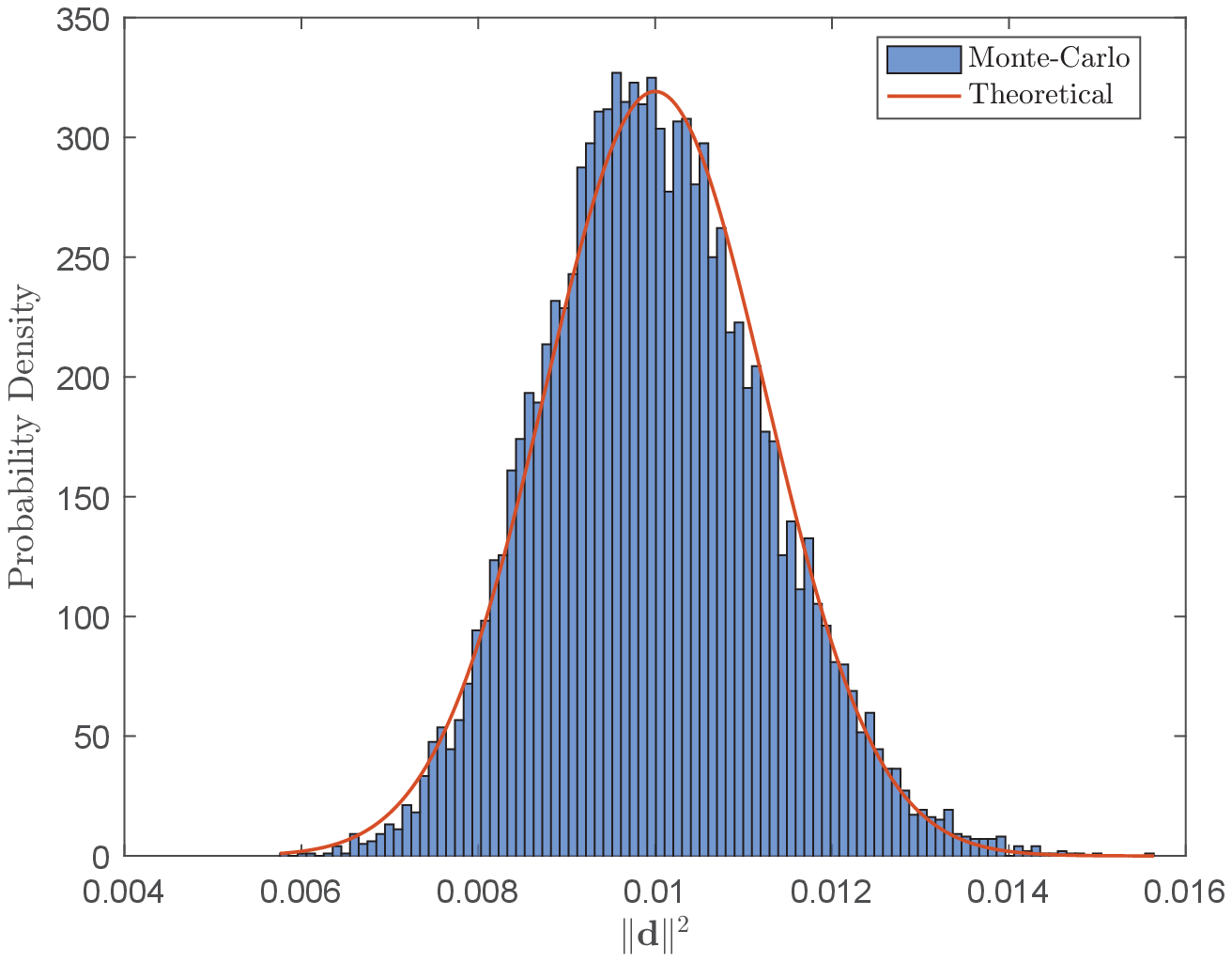}}\quad
\subfigure[{\it p.d.f.} of $\|\mathbf{d}+\mathbf{G\Phi f}\|^2$]{%
\epsfxsize=0.2\textwidth \leavevmode
\epsffile{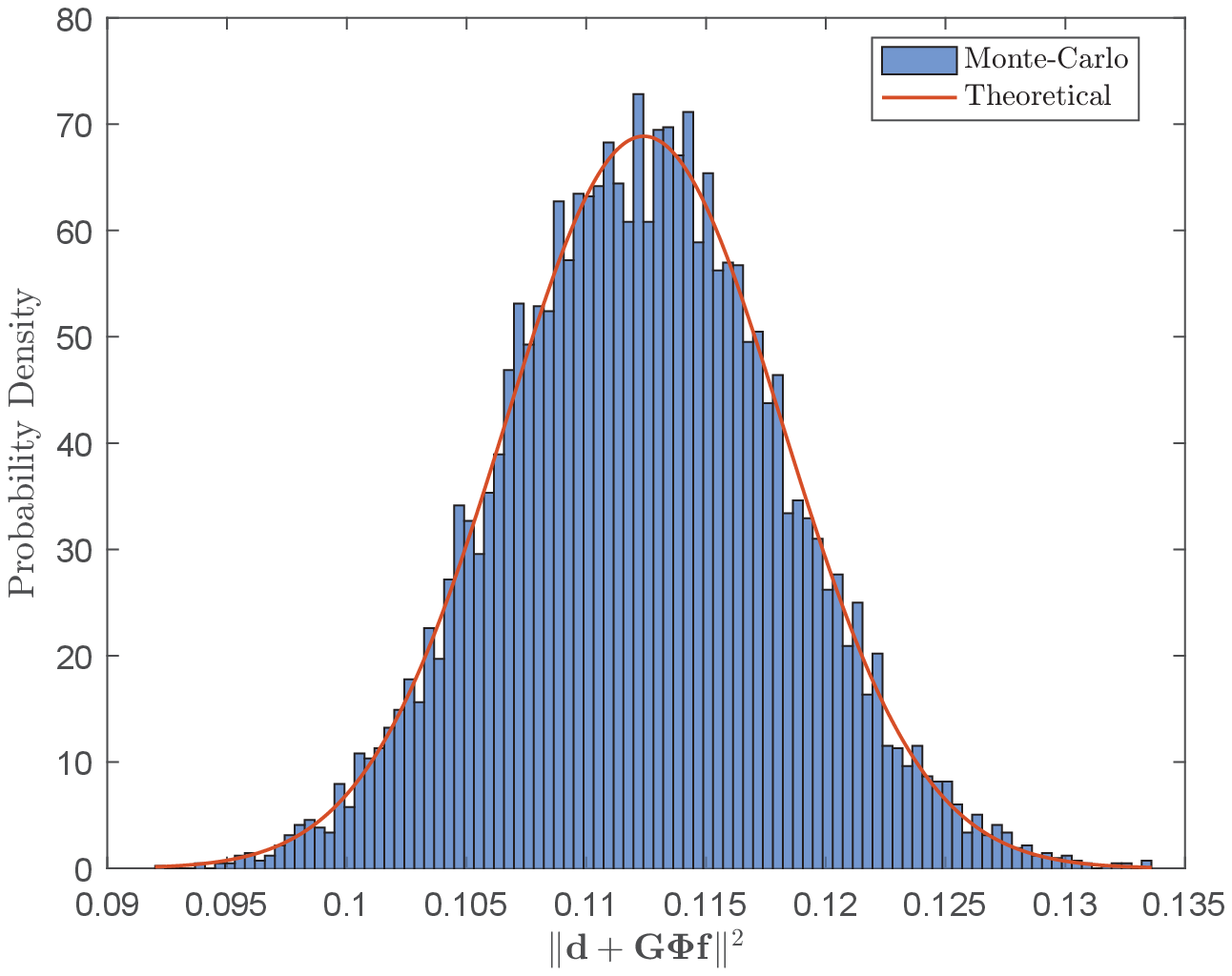}}
\caption{Analytical results and Monte-Carlo simulation results for the {\it p.d.f.} of $\|\mathbf{d}\|^2$ and $\|\mathbf{d}+\mathbf{G\Phi f}\|^2$ under LoS RIS-related channels.}\label{fig:pdfvad_LoS}
\end{center}
\end{figure}

\subsection{Under Rayleigh RIS-related Channels}\label{subsec:Rayleigh}
In practice, the RIS-related channels are more likely to have LoS components, and therefore the statistical phase shift design is expected to be useful. In this section, we still consider the Rayleigh case where the RIS-related channels are of Rayleigh fading to provide more insights into the statistical phase shift design, since the Rayleigh case is also a special case of the Rician case we will introduce later. Under the Rayleigh case, the statistical phase shift becomes actually random phase shift design since both $\mathbf{f}$ and $\mathbf{G}$ are random. Hence, we provide the analysis under where $\mathbf{\Phi}$ is randomly chosen. Denoting $r\triangleq \|\mathbf{G}\mathbf{\Phi}\mathbf{f}\|^2$, for large $M$ and $N$, the distribution of $r$ converges to a Gaussian distribution as follows.
\begin{lemma}\label{lem:rdfRay}
Let $\mathbf{f}\in\mathbb{C}^{N}$ with $\mathbf{f}\sim\mathcal{CN}(\mathbf{0}, \beta_f\mathbf{I}_M)$, $\mathbf{G}\in\mathbb{C}^{N\times M}$ with {\it i.i.d.} entries $G_{ij}\sim\mathcal{CN}(0, \beta_G)$, $\mathbf{\Phi}\in\mathbb{C}^{M\times M}$ is a random diagonal matrix. For large $M$, $N$, $r=\|\mathbf{G}\mathbf{\Phi}\mathbf{f}\|^2$ converges to a Gaussian distribution as
\begin{equation}\label{eq:rdfRay}
r\sim\mathcal{N}\left(\mu_r^{\rm Ray}, v_r^{\rm Ray}\right),
\end{equation}
where
\begin{equation}\label{eq:mrRay}
\mu_r^{\rm Ray} = MN\beta_f\beta_G, \ v_r^{\rm Ray} = (M+N)MN\beta_f^2\beta_G^2.
\end{equation}
%\begin{equation}\label{eq:vrRay}
%v_r^{\rm Ray} = (M+N)MN\beta_f^2\beta_G^2.
%\end{equation}

\end{lemma}
\begin{proof}
See Appendix \ref{apdx:prooflem1}.
\end{proof}

With the CLT, we can infer that $g$ also satisfies a Gaussian distribution. Since $\mathbf{d}$ is independent of $\mathbf{G}\mathbf{\Phi}\mathbf{f}$, we have
\begin{equation}\label{eq:expectgRay}
\mathbb{E}[g]= \mathbb{E}[\|\mathbf{d}+\mathbf{G}\mathbf{\Phi}\mathbf{f}\|^2]
=\mathbb{E}[\|\mathbf{d}\|^2]+\mathbb{E}[\|\mathbf{G}\mathbf{\Phi}\mathbf{f}\|^2]=\mu_d+\mu_r^{\rm Ray}.
\end{equation}
To derive the variance of $g$, we first approximate the cascaded channel, i.e, $\mathbf{G}\mathbf{\Phi}\mathbf{f}$, with a Rayleigh fading channel $\mathbf{r}\sim\mathcal{CN}(\mathbf{0}, \beta_r\mathbf{I}_N)$ since $\mathbf{f}$, $\mathbf{G}$ are Rayleigh fading channels and $\mathbf{\Phi}$ is randomly chosen. Using the fact that $\|\mathbf{r}\|^2$ and $\|\mathbf{G}\mathbf{\Phi}\mathbf{f}\|^2$ have the same variance, we can obtain
\begin{equation}\label{eq:beta_r}
\beta_r=\sqrt{\frac{v_r^{\rm Ray}}{N}}=\sqrt{(M+N)M}\beta_f\beta_G.
\end{equation}
Then, we approximate $\mathbf{d}+\mathbf{G}\mathbf{\Phi}\mathbf{f}$ with the summation of two independent Rayleigh channels, namely, $\mathbf{d}$ and $\mathbf{r}$, when we derive the variance of $g$. Via the CLT, the variance of $g$ is
\begin{equation}\label{eq:vargRay}
v_g^{\rm Ray} =N(\beta_d+\beta_r)^2= N(\beta_d+\sqrt{(M+N)M}\beta_f\beta_G)^2.
\end{equation}
Therefore, the asymptotic distribution of $g$ under the Rayleigh case unfolds as the following proposition.
\begin{prop}\label{prop:gdfRay}
For large $M$, $N$, the distribution of $g$ under the Rayleigh case is given by
\begin{equation}\label{eq:gdfRay}
g\sim\mathcal{N}(\mu_g^{\rm Ray}, v_g^{\rm Ray}),
\end{equation}
where
\begin{equation}\label{eq:mgRay}
\mu_g^{\rm Ray} = N\beta_d+MN\beta_G\beta_f, \ v_g^{\rm Ray} = N(\beta_d+\sqrt{(M+N)M}\beta_f\beta_G)^2.
\end{equation}
%\begin{equation}\label{eq:vgRay}
%v_g^{\rm Ray} = N(\beta_d+\sqrt{(M+N)M}\beta_f\beta_G)^2.
%\end{equation}
\end{prop}

\begin{figure}[!t]
\begin{center}
\subfigure[{\it p.d.f.} of $\|\mathbf{G\Phi f}\|^2$]{%
\epsfxsize=0.2\textwidth \leavevmode
\epsffile{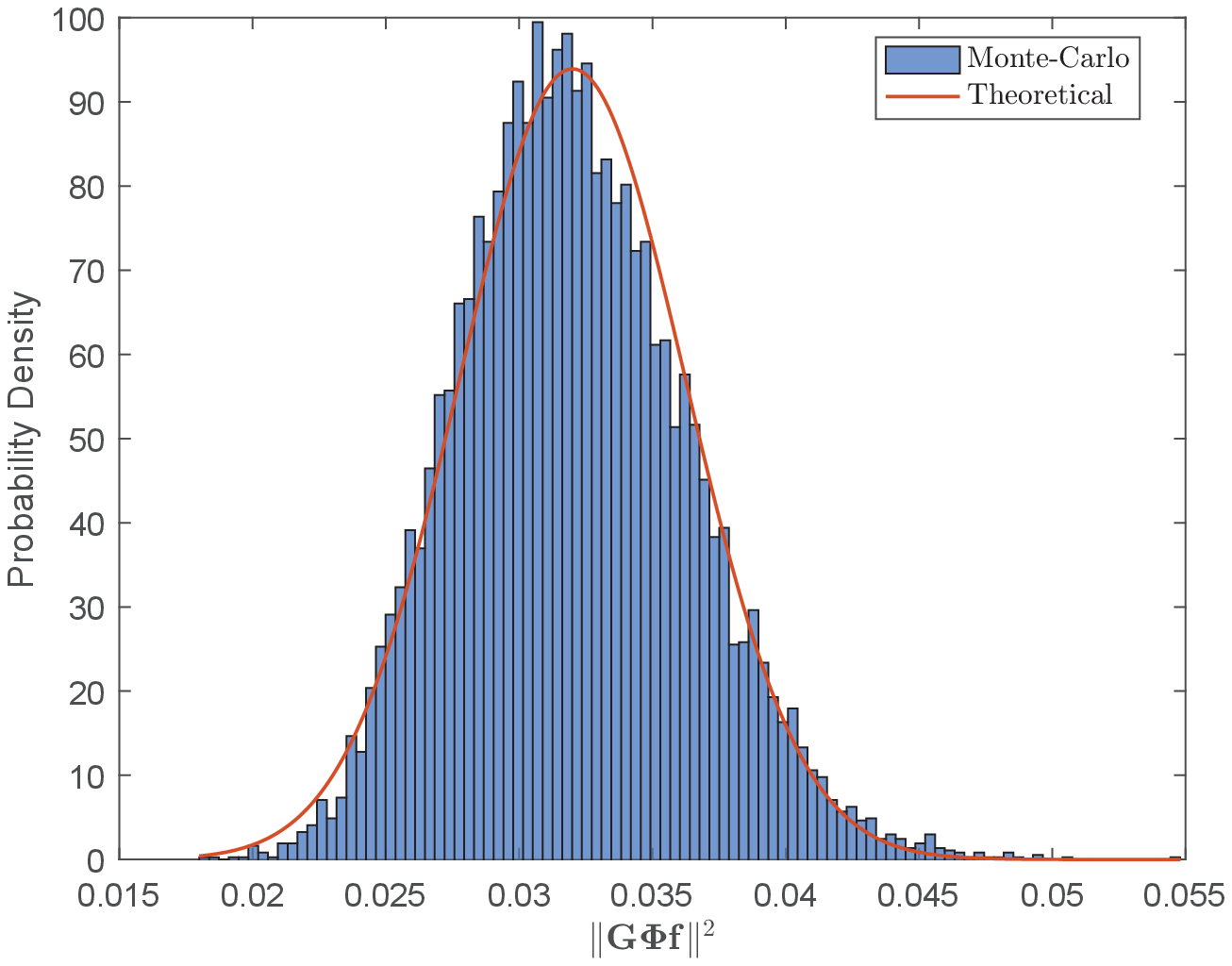}}\quad
\subfigure[{\it p.d.f.} of $\|\mathbf{d}+\mathbf{G\Phi f}\|^2$]{%
\epsfxsize=0.2\textwidth \leavevmode
\epsffile{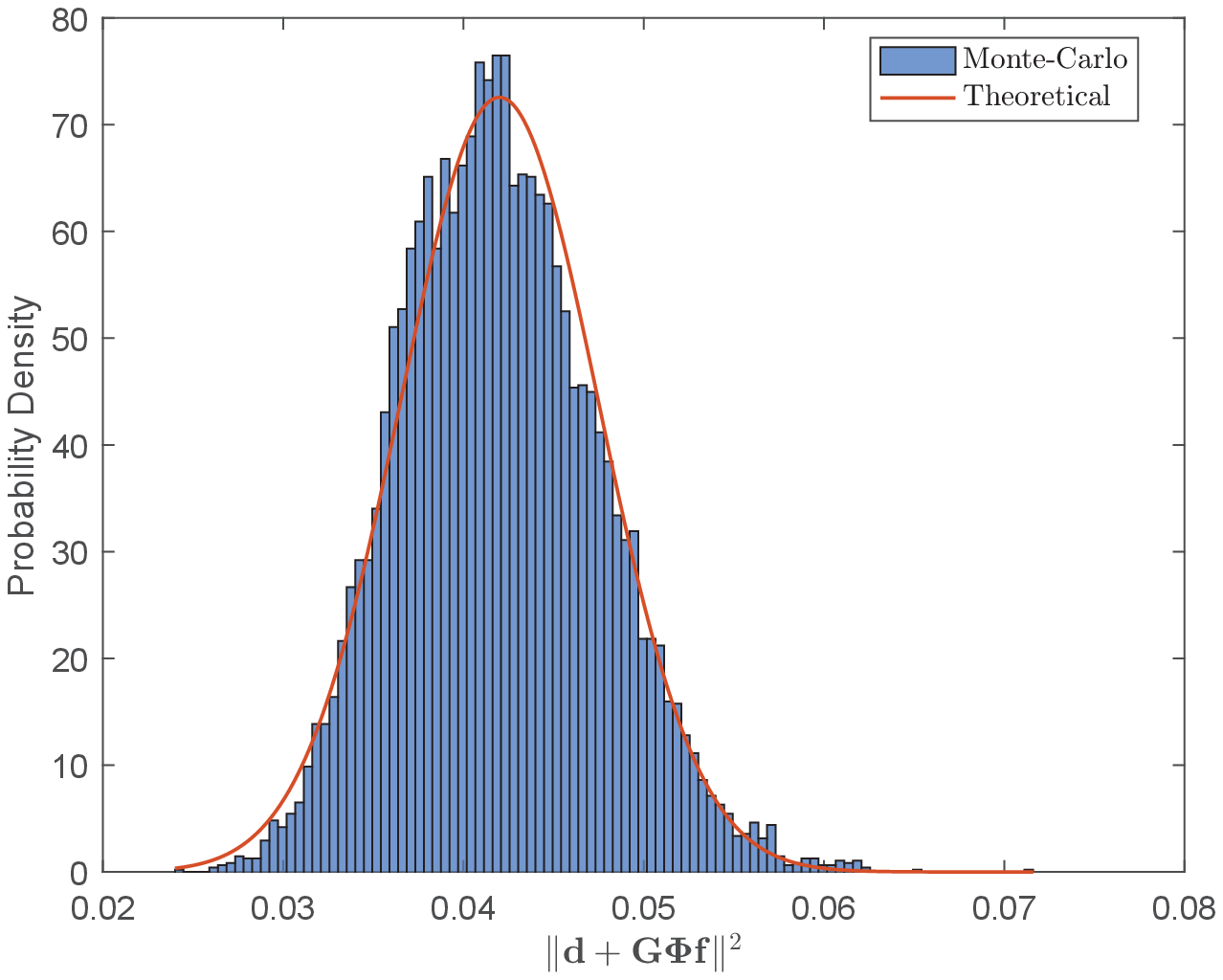}}\quad
\caption{Analytical results and Monte-Carlo simulation results for the {\it p.d.f.} of $\|\mathbf{G\Phi f}\|^2$ and $\|\mathbf{d}+\mathbf{G\Phi f}\|^2$ under Rayleigh RIS-related channels.}\label{fig:pdfvad_Ray}
\end{center}
\end{figure}

Finally, we perform the Monte-Carlo simulations with $N\beta_d=-20dB, N=64, M=500, \beta_f\beta_G=-60dB$ to evaluate the accuracy of the derived asymptotic distribution of $g$. The results are shown in Fig. \ref{fig:pdfvad_Ray} and the analytical results fit the simulation results well.

\subsection{Under Rician RIS-related Channels}\label{subsec:Rician}
Under the LoS case, the statistical phase shift design provides us a reliable cascaded channel to enhance the equivalent channel gain. On the other hand, under the Rayleigh case, if the phase shift matrix is not adaptively tuned according to the instantaneous CSI, the gain introduced by the RIS is much lower than the LoS case. This can be easily observed by comparing the mean value in \eqref{eq:mgLoS} and that in \eqref{eq:mgRay}. In a word, the LoS RIS-related channels are more suitable for deploying the statistical phase shift design. However, the LoS RIS-related channels can not be always achieved due to the practical placement of the RIS. Hence, we here consider the more general Rician case where $\mathbf{f}$ and $\mathbf{G}$ are of Rician fading.

Let
\begin{equation}\label{eq:ricianf}
\mathbf{f} = \sqrt{\frac{\beta_f\kappa_f}{\kappa_f+1}}\bar{\mathbf{f}}+\sqrt{\frac{\beta_f}{\kappa_f+1}}\tilde{\mathbf{f}},
\mathbf{G} = \sqrt{\frac{\beta_G\kappa_G}{\kappa_G+1}}\bar{\mathbf{G}}+\sqrt{\frac{\beta_G}{\kappa_G+1}}\tilde{\mathbf{G}},
\end{equation}
%\begin{equation}\label{eq:ricianG}
%\mathbf{G} = \sqrt{\frac{\beta_G\kappa_G}{\kappa_G+1}}\bar{\mathbf{G}}+\sqrt{\frac{\beta_G}{\kappa_G+1}}\tilde{\mathbf{G}},
%\end{equation}
where $\bar{\mathbf{f}}=\mathbf{a}_{f}$, $\bar{\mathbf{G}}=\mathbf{a}_{G}\mathbf{b}_{G}^H$ with $\mathbf{a}_{f}$, $\mathbf{a}_{G}$, and $\mathbf{b}_{G}$ defined in \eqref{eq:LoSfsv}, \eqref{eq:LoSGsva} and \eqref{eq:LoSGsvb}, respectively; $\tilde{\mathbf{f}}\sim\mathcal{CN}(\mathbf{0}, \mathbf{I}_M)$, $\tilde{\mathbf{G}}$ has {\it i.i.d.} entries with $\tilde{\mathbf{G}}_{ij}\sim\mathcal{CN}(0, 1)$; $\kappa_f$ and $\kappa_G$ are the Rician factors of $\mathbf{f}$ and $\mathbf{G}$, respectively. Again, we consider to statistically determine $\mathbf{\Phi}$ by solving the optimization problem in \eqref{eq:optimalPhiLoS}. It can be imagined that the optimal phase shift matrix of RIS is still the solution in \eqref{eq:optimalPSMLoS}. With the optimized $\mathbf{\Phi}$, for large $M$, $N$, the distribution of $r=\|\mathbf{G}\mathbf{\Phi}\mathbf{f}\|^2$ converges to a Gaussian distribution as follows.

\begin{lemma}\label{lem:rdfRi}
Denoting
\begin{equation}\label{eq:betafbt}
\bar{\beta}_f\triangleq\frac{\beta_f\kappa_f}{\kappa_f+1},\  \tilde{\beta}_f\triangleq\frac{\beta_f}{\kappa_f+1},
\bar{\beta}_G\triangleq\frac{\beta_G\kappa_G}{\kappa_G+1}, \tilde{\beta}_G\triangleq\frac{\beta_G}{\kappa_G+1},
\end{equation}
%\begin{equation}\label{eq:betaGbt}
%\bar{\beta}_G\triangleq\frac{\beta_G\kappa_G}{\kappa_G+1}, \tilde{\beta}_G\triangleq\frac{\beta_G}{\kappa_G+1},
%\end{equation}
$r=\|\mathbf{G}\mathbf{\Phi}\mathbf{f}\|^2$ converges to a Gaussian distribution as
\begin{equation}\label{eq:rdfRi}
r\sim\mathcal{N}(\mu_r^{\rm Ri}, v_r^{\rm Ri}),
\end{equation}
where
\begin{equation}\label{eq:mrRi}
\mu_r^{\rm Ri}=M^2N\bar{\beta}_f\bar{\beta}_G+MN(\tilde{\beta}_f\bar{\beta}_G+\bar{\beta}_f\tilde{\beta}_G+\tilde{\beta}_f\tilde{\beta}_G),
\end{equation}
$v_r^{\rm Ri}$ is defined as
\begin{align}\label{eq:vrRi}
v_r^{\rm Ri}=&2M^3N^2\bar{\beta}_f\tilde{\beta}_f\bar{\beta}_G^2+2M^3N\bar{\beta}_f^2\bar{\beta}_G\tilde{\beta}_G+
(2M^3N+4M^2N^2+4M^2N)\bar{\beta}_f\tilde{\beta}_f\bar{\beta}_G\tilde{\beta}_G\nonumber\\
&+M^2N^2\tilde{\beta}_f^2\bar{\beta}_G^2+2M^2N\tilde{\beta}_f^2\bar{\beta}_G\tilde{\beta}_G+
2(M+N)MN\bar{\beta}_f\tilde{\beta}_f\tilde{\beta}_G^2+(M+N)MN\tilde{\beta}_f^2\tilde{\beta}_G^2.
\end{align}
\end{lemma}
\begin{proof}
See Appendix \ref{apdx:prooflem2}.
\end{proof}

It is worth noting that Lemma \ref{lem:rdfRay} is actually a special case of Lemma \ref{lem:rdfRi} when $\kappa_f=\kappa_G=0$. With the statistically determined $\mathbf{\Phi}$, it can be observed that $\bar{\mathbf{f}}$, $\mathbf{\Phi}$, and $\bar{\mathbf{G}}$ constitute an equivalent LoS component of $\mathbf{G}\mathbf{\Phi}\mathbf{f}$. Thus, we can approximate $\mathbf{G}\mathbf{\Phi}\mathbf{f}$ with an equivalent Rician fading channel, namely, $\mathbf{h}=\bar{\mathbf{h}}+\tilde{\mathbf{h}}$, which satisfies $\mathbf{h}\sim\mathcal{CN}(\bar{\mathbf{h}}, \sigma_h^2\mathbf{I}_N)$. For large $N$, with the CLT, $h\triangleq\|\mathbf{h}\|^2$ converges to a Gaussian distribution as
\begin{equation}\label{eq:hdfRi}
h\sim\mathcal{N}(N\sigma_h^2+\|\bar{\mathbf{h}}\|^2, N\sigma_h^4+2\|\bar{\mathbf{h}}\|^2\sigma_h^2).
\end{equation}
Comparing \eqref{eq:hdfRi} with \eqref{eq:rdfRi}, the corresponding equivalent parameters, i.e., $\|\bar{\mathbf{h}}\|^2$ and $\sigma_h^2$, can be obtained via
\begin{equation}\label{eq:emh}
\|\bar{\mathbf{h}}\|^2=\sqrt{N\left(\frac{{\mu_r^{\rm Ri}}^2}{N}-v_r^{\rm Ri}\right)},\
\sigma_h^2=\frac{\mu_r^{\rm Ri}-\|\bar{\mathbf{h}}\|^2}{N}.
\end{equation}
%\begin{equation}\label{eq:evh}
%\sigma_h^2=\frac{\mu_r^{\rm Ri}-\|\bar{\mathbf{h}}\|^2}{N}.
%\end{equation}
With this approximation, $\mathbf{d}+\mathbf{G}\mathbf{\Phi}\mathbf{f}$ can also be approximated with a Rician fading channel whose LoS component is exactly $\bar{\mathbf{h}}$. Besides, with the fact that $\tilde{\mathbf{h}}$ is independent of $\mathbf{d}$, we have
\begin{equation}\label{eq:haddd}
\tilde{\mathbf{h}}+\mathbf{d}\sim\mathcal{CN}\left(\mathbf{0}, (\sigma_h^2+\beta_d)\mathbf{I}_N\right).
\end{equation}
Finally, $\mathbf{d}+\mathbf{G}\mathbf{\Phi}\mathbf{f}$ can be approximated by $\bar{\mathbf{h}}+\tilde{\mathbf{h}}+\mathbf{d}$.
Obviously, $\mathbf{g}\triangleq \mathbf{d}+\bar{\mathbf{h}}+\tilde{\mathbf{h}}$ is another Rician fading channel which satisfies $\mathbf{g}\sim\mathcal{CN}(\bar{\mathbf{h}}, (\sigma_h^2+\beta_d)\mathbf{I}_N)$. Thus, the asymptotic distribution of $g$ under the Rician case unfolds as the following proposition.
\begin{prop}\label{prop:gdfRi}
For large $M$, $N$, with the CLT, the distribution of $g=\|\mathbf{g}\|^2$ under the Rician case is given by
\begin{equation}\label{eq:gdfRi}
g\sim\mathcal{N}(\mu_g^{\rm Ri}, v_g^{\rm Ri}),
\end{equation}
where
\begin{equation}\label{eq:mgRi}
\mu_g^{\rm Ri} = N(\sigma_h^2+\beta_d)+\|\bar{\mathbf{h}}\|^2, \ v_g^{\rm Ri} = N(\sigma_h^2+\beta_d)^2+2\|\bar{\mathbf{h}}\|^2(\sigma_h^2+\beta_d).
\end{equation}
%\begin{equation}\label{eq:vgRi}
%v_g^{\rm Ri} = N(\sigma_h^2+\beta_d)^2+2\|\bar{\mathbf{h}}\|^2(\sigma_h^2+\beta_d).
%\end{equation}
\end{prop}

To validate the derived asymptotic distribution of $g$, the Monte-Carlo simulation results for $N\beta_d=-20dB, N=64, M=40, \beta_f\beta_G=-60dB$ are provided in Fig. \ref{fig:pdfvad_Ri} and are consistent with the analytical results.

\begin{figure}[!t]
\begin{center}
\subfigure[{\it p.d.f.} of $\|\mathbf{G\Phi f}\|^2$]{%
\epsfxsize=0.2\textwidth \leavevmode
\epsffile{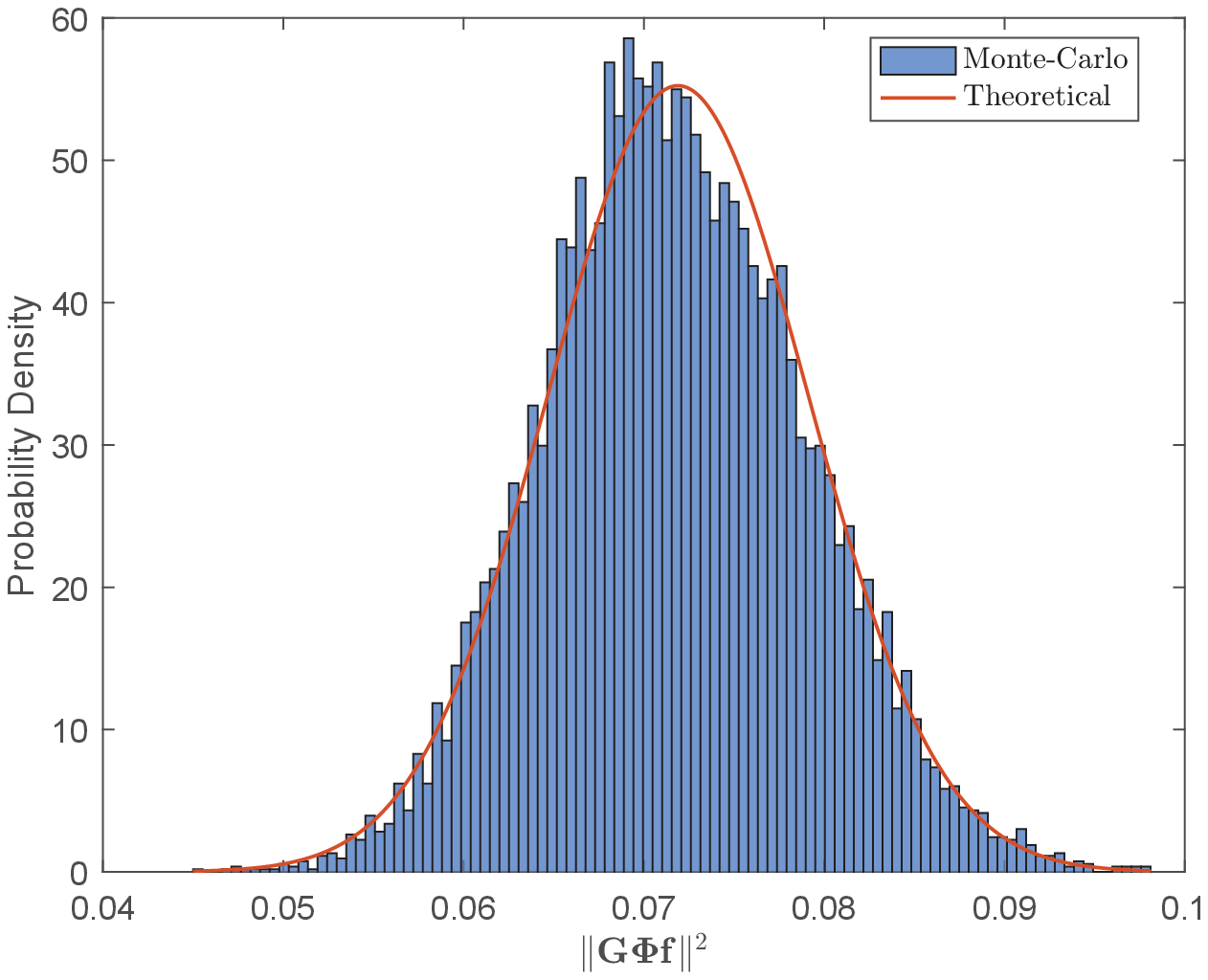}}\quad
\subfigure[{\it p.d.f.} of $\|\mathbf{d}+\mathbf{G\Phi f}\|^2$]{%
\epsfxsize=0.2\textwidth \leavevmode
\epsffile{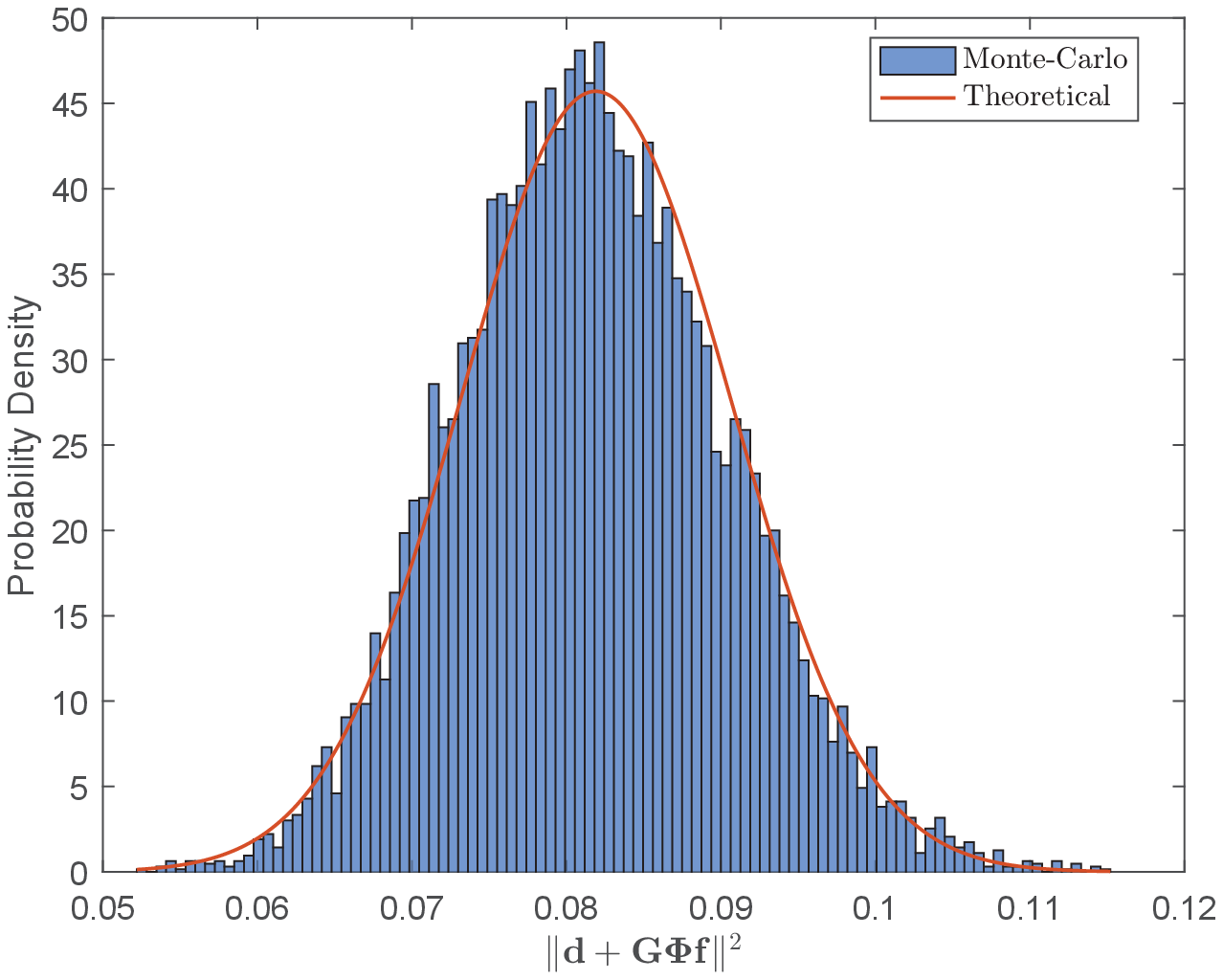}}\quad
\caption{Analytical results and Monte-Carlo simulation results for the {\it p.d.f.} of $\|\mathbf{G\Phi f}\|^2$ and $\|\mathbf{d}+\mathbf{G\Phi f}\|^2$ under Rician RIS-related channels.}\label{fig:pdfvad_Ri}
\end{center}
\end{figure}

%we can obtain the asymptotic distribution of $g$, i.e., $\mathcal{N}(\mu_g, v_g)$, by
%$$\mu_g=N(\sigma_h^2+\beta_d)+\|\bar{\mathbf{h}}\|^2,$$
%$$v_g=N(\sigma_h^2+\beta_d)^2+2\|\bar{\mathbf{h}}\|^2(\sigma_h^2+\beta_d).$$

\section{How Can We Achieve a Detection Probability Close to $1$?}\label{sec:perfectdetection}

In Section \ref{sec:CGA}, we have discussed the statistical phase shift design and have derived the asymptotic distributions of the equivalent channel gains for three cases of the RIS-related channels. Particularly, all the distributions are shown to be asymptotically Gaussian due to the CLT. However, it is still unclear that how to place the RIS or how many REs we need to achieve a significant improvement of the spectrum sensing performance. According to Remark \ref{rmk:spikedetection}, if the probability of $g<\sqrt{c}$ is almost $1$, the asymptotic distributions of the test statistic under $\mathcal{H}_0$ and $\mathcal{H}_1$ are the same and thus the ST can not identify the active PT. Therefore, based on the relation between $\sqrt{c}$ and the asymptotic distribution of $g$, we may have three possible cases, which are illustrated in Fig. \ref{fig:perfectdetection}.

In Fig. \ref{fig:perfectdetection}, we use the red-shaded part and the blue-shaded part to denote the probability of $g<\sqrt{c}$ and that of $g>\sqrt{c}$, respectively. Besides, more details for the three cases are listed as follows.
\begin{figure*}[!t]
\begin{center}
\subfigure[Case $1$]{%
\epsfxsize=0.2\textwidth \leavevmode
\epsffile{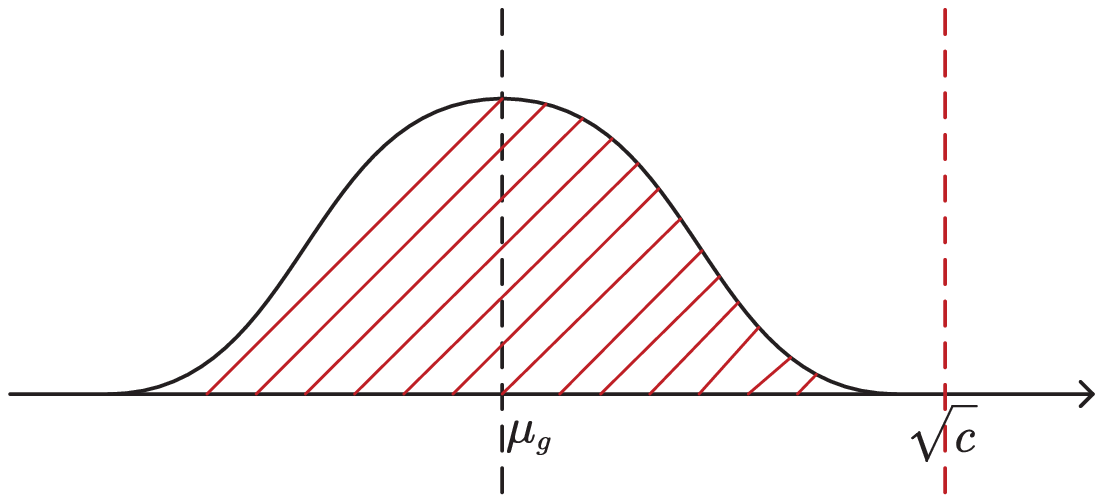}}\qquad
\subfigure[Case $2$]{%
\epsfxsize=0.2\textwidth \leavevmode
\epsffile{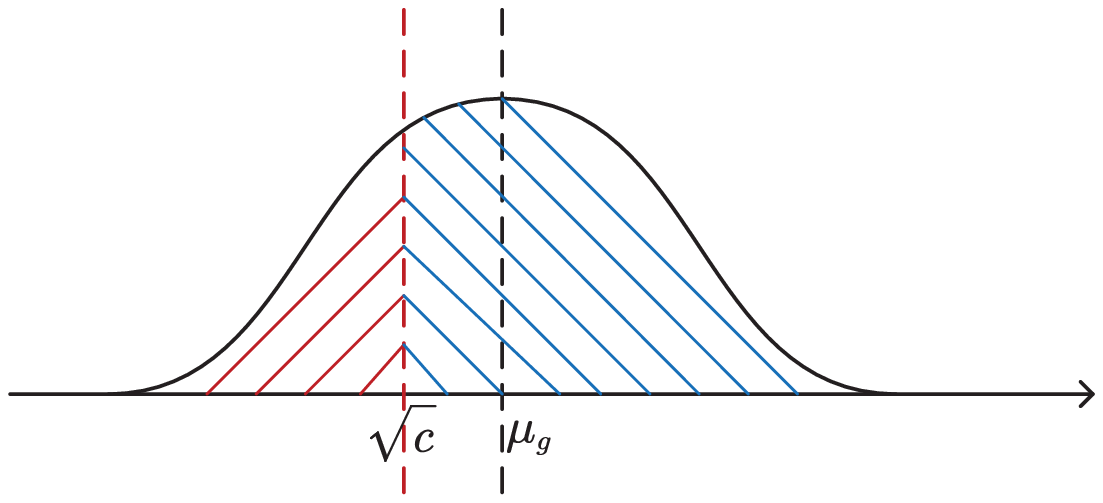}}\qquad
\subfigure[Case $3$]{%
\epsfxsize=0.2\textwidth \leavevmode
\epsffile{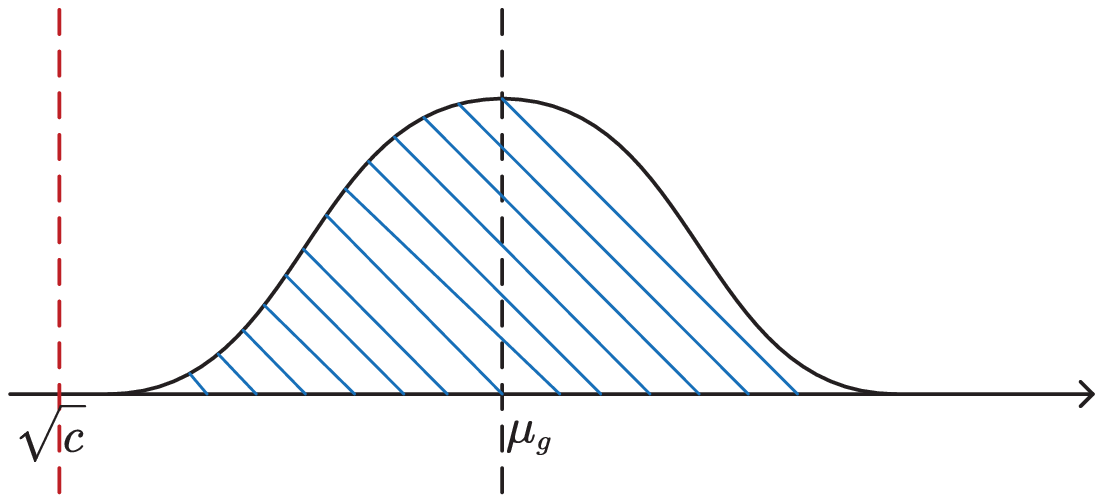}}
\caption{Three possible cases for the asymptotic distribution of $g$.}\label{fig:perfectdetection}
\end{center}
\end{figure*}

\begin{itemize}
  \item {\it Case} $1$: the probability of $g<\sqrt{c}$ is almost $1$, thus the {\it p.d.f.}'s of $T$ under $\mathcal{H}_0$ and $\mathcal{H}_1$ are almost surely the same, we have $P_{d}=P_{fa}$.
  \item {\it Case} $2$: the red-shaded part denotes the probability of that the {\it p.d.f.} of $T$ under $\mathcal{H}_0$ and $\mathcal{H}_1$ are the same, i.e., $P_{d}=P_{fa}$, and the blue-shaded part denotes the probability of that $T$ under $\mathcal{H}_1$ satisfies a spiked distribution as shown in \eqref{eq:gdfgeq}.
  \item {\it Case} $3$: the probability of $g>\sqrt{c}$ is almost $1$, thus the {\it p.d.f.} of $T$ under $\mathcal{H}_1$ almost surely converges to a spiked distribution as shown in \eqref{eq:gdfgeq}.
\end{itemize}

Obviously, if $g<\sqrt{c}$ with a nonzero probability, e.g., {\it Case} $1$ and {\it Case} $2$, the {\it miss detection} event will almost surely happens with a nonzero probability. To achieve a detection probability close to $1$, i.e., $P_{d}\approx 1$, the probability of $g<\sqrt{c}$ are supposed to be as small as possible. Thus, {\it Case} $3$ is exactly a necessary condition to achieve a detection probability close to $1$, which means that, if $P_{d}\approx 1$ is achieved, the relation between the asymptotic distribution of $g$ and $\sqrt{c}$ must be as in {\it Case} $3$. With the results in Section \ref{sec:CGA}, we know that the asymptotic distributions of $g$ under all the channel conditions converge to some specific Gaussian distributions. Without loss of generality, we let $g\sim\mathcal{N}(\mu_g, v_g)$ with $v_g=\sigma_g^2$. Using the {\it Three-Sigma Rule} for Gaussian distributions, we have
\begin{equation}\label{eq:g3s}
P(g>\mu_g-3\sigma_g)\approx 0.9987.
\end{equation}
Therefore, {\it Case} $3$ can be approximately guaranteed by letting
\begin{equation}\label{eq:nececonditionPD}
\mu_g-3\sigma_g>\sqrt{c},
\end{equation}
which can effectively achieve a near-zero probability of $g<\sqrt{c}$, or more accurately, $P(g<\sqrt{c})<0.0013$. Besides, $\mu_g$ and $\sigma_g$ are actually functions of $M$, e.g., \eqref{eq:mgRay}, let $g_m(M)\triangleq\mu_g-3\sigma_g$ and $M_{\inf}$ be the solution to the equation $g_m(M)=\sqrt{c}$. With the fact that better detection performance can be achieved by employing more REs, we can see that $M_{\inf}$ is actually the lower bound of $M$ to achieve a detection probability close to $1$.

To achieve $P_{d}\approx 1$, we have to employ an RIS with more than $M_{\inf}$ REs. However, it is still unknown that how many REs are required. Hence, we have to reconsider {\it Case} $3$ where $g$ is almost surely larger than $\sqrt{c}$. Under {\it Case} $3$, we recall that the {\it p.d.f.} of $T$ under $\mathcal{H}_1$ converges to a spiked distribution with mean $\mu_T(g)$ and variance $v_T(g)=\sigma_T^2(g)$ as shown in \eqref{eq:gdfgeq}. Therefore, $P_{d}\approx 1$ can be achieved by letting $T>\gamma$ almost surely happens under $\mathcal{H}_1$. Using the {\it Three-Sigma Rule} again, we have
\begin{equation}\label{eq:t3s}
P\left(T>\mu_T(g)-3\sigma_T(g)|g\right)\approx0.9987, \forall g>g_0>\sqrt{c}.
\end{equation}
Thus, the sufficient condition to achieve $P_{d}\approx 1$ can be expressed as
\begin{equation}\label{eq:suffconditionPD1}
\mu_T(g)-3\sigma_T(g)>\gamma, \forall g>g_0.
\end{equation}

To make the sufficient condition more clear, we first evaluate the monotonicity of $\mu_T(g)-3\sigma_T(g)$ with respect to $g$.
Let $f(g)\triangleq\mu_T(g)-3\sigma_T(g)$, with \eqref{eq:mTspike}, we have
\begin{equation}\label{eq:fg}
f(g)=g+1+c+\frac{c}{g}-3\sqrt{\frac{(g+1)^2}{n}\left(1-\frac{c}{g^2}\right)},
\end{equation}
where $g$ is defined in the interval $(\sqrt{c}, +\infty)$. Take the derivative of $f(g)$ with respect to $g$, we obtain
\begin{equation}\label{eq:fdg}
f'(g)=\frac{(g^2-c)\left(\sqrt{n(g^2-c)}+3\right)-3g^2(g+1)}{g^2\sqrt{n(g^2-c)}}.
\end{equation}
It can be easily observed that the denominator of $f'(g)$ are always positive when $g>\sqrt{c}$, therefore the sign of $f'(g)$ only depends on the numerator, which is denoted by $f'_{num}$ in the context. Denoting $(\sqrt{c})^{+}$ a number which is a little larger than $\sqrt{c}$ but its limit is $\sqrt{c}$, i.e.,
$\lim(\sqrt{c})^{+}=\sqrt{c}$, we have
\begin{equation}\label{eq:fdgsqrtc}
\lim f'_{num}[(\sqrt{c})^{+}]=-3c(\sqrt{c}+1)<0.
\end{equation}
Besides, consider an arbitrary value of $g$, namely, $m\sqrt{c}, m>1$, we have
\begin{equation}\label{eq:fdgmsqrtc}
f'_{num}(m\sqrt{c})=\left(\left[(m^2-1)^{\frac{3}{2}}\sqrt{n}-3m^3\right]\sqrt{c}-3\right)c.
\end{equation}
For large $m$, $(m^2-1)^{\frac{3}{2}}$ can be approximated by $m^3$, then we have $f'_{num}(m\sqrt{c})=\left[(\sqrt{n}-3)m^3\sqrt{c}\right.$ $\left.-3\right]c$. Moreover, the number of signal samples, namely, $n$, is much larger than $3^2$ and note that $\lim c=N/n$, thus $f'_{num}(m\sqrt{c})$ can be further approximated by $(\sqrt{N}m^3-3)c$, where we recall $N$ is the number of the antennas at ST. To conclude, $f'(m\sqrt{c})$ is negative when $m$ is smaller than a certain threshold $m_0$, and is always positive when $m$ becomes larger than $m_0$. This means that $f(g)$ decreases firstly and then increases as $g$ increases. Therefore, if the equation $f(g)=\gamma$ has two roots, \eqref{eq:suffconditionPD1} can be achieved by letting $g_0$ be the larger root. Besides, if the equation has a unique root or no root, $M_{\inf}$ becomes sufficient to achieve a detection probability close to $1$.

Via some preliminary numerical simulations, we find that the equation $f(g)=\gamma$ almost always has two roots when a reasonable false alarm probability is considered. Let $g_0$ be the larger root of the equation $f(g)=\gamma$, denoting $M_{\rm PD}$ the solution to the equation $g_m(M)=g_0$, we finally obtain the number of REs required to achieve a detection probability close to $1$. Obviously, $M_{\rm PD}$ is large enough to achieve $P_{d}\approx 1$ since it makes \eqref{eq:nececonditionPD} and \eqref{eq:suffconditionPD1} hold at the same time. Hence, $M_{\rm PD}$ actually provides us a theoretical prediction about the number of REs required to achieve $P_{d}\approx 1$.

\section{Simulation Results}\label{sec:sim}

In this section, we provide numerical simulations to compare the analytical results with the Monte-Carlo simulation results and evaluate the validity of the theoretical predictions about the number of REs required to achieve $P_{d}\approx 1$. In addition, we investigate the impact of the channel Rician factors on the number of REs required to achieve a high detection probability.

\subsection{Simulation Setup}\label{subsec:simsetup}

It is well-known that the detection performance mainly depends on the number of signal samples and the SNR for sensing. In conventional spectrum sensing systems, the SNR can be defined as \cite{bianchi2011performance}
\begin{equation}\label{eq:sensingSNR}
SNR=\frac{\sigma_s^2\mathbb{E}[\|\mathbf{d}\|^2]}{\sigma_u^2}=N\beta_d,
\end{equation}
where we recall that $\mathbf{d}$ denotes the channel from PT to ST. In the RIS-enhanced spectrum sensing system, $\mathbf{d}$ should be substituted with $\mathbf{d}+\mathbf{G}\mathbf{\Phi}\mathbf{f}$.
%Besides, we use $\mathbf{g}^{\rm w}=\mathbf{d}+\mathbf{G}\mathbf{\Phi}\mathbf{f}$ to denote the equivalent channel in the RIS-enhanced spectrum sensing system. On the contrary, $\mathbf{g}^{\rm wo}=\mathbf{d}$ denotes the equivalent channel when RIS is not employed.

In CR, the secondary user is supposed to detect the presence of the active PT with a high detection probability at a quite low SNR level, say, $-20dB$ \cite{zeng2010review}. Hence, we set $N\beta_d=-20dB$ in the following simulations. On the other hand, the cascaded channel, i.e., $\mathbf{G\Phi f}$, usually suffers from the double fading effect \cite{griffin2009complete}, i.e, the pathloss of the cascaded channel is much lower than that of the direct-link channel $\mathbf{d}$. The specific value of the pathloss of the cascaded channel mainly depends on the placement of the RIS, we here set $\beta_f\beta_G=-60dB$ without loss of generality. Besides, we set the false alarm probability, i.e., $\alpha$, to $0.1$. The detection threshold can therefore be calculated by \eqref{eq:gamma}, in which $F_2^{-1}(1-\alpha)$ can be computed with the tools provided in \cite{RMTstat}.

\subsection{Accuracy of the Asymptotic Analytical Framework for Evaluating $P_d$ in the Finite Dimensional Regime}\label{subsec:Pdtheoaccu}

\begin{figure*}[!t]
\setlength{\abovecaptionskip}{3pt}
\centering
\subfigure[$N=64$]{%
\epsfxsize=0.24\textwidth \leavevmode
\epsffile{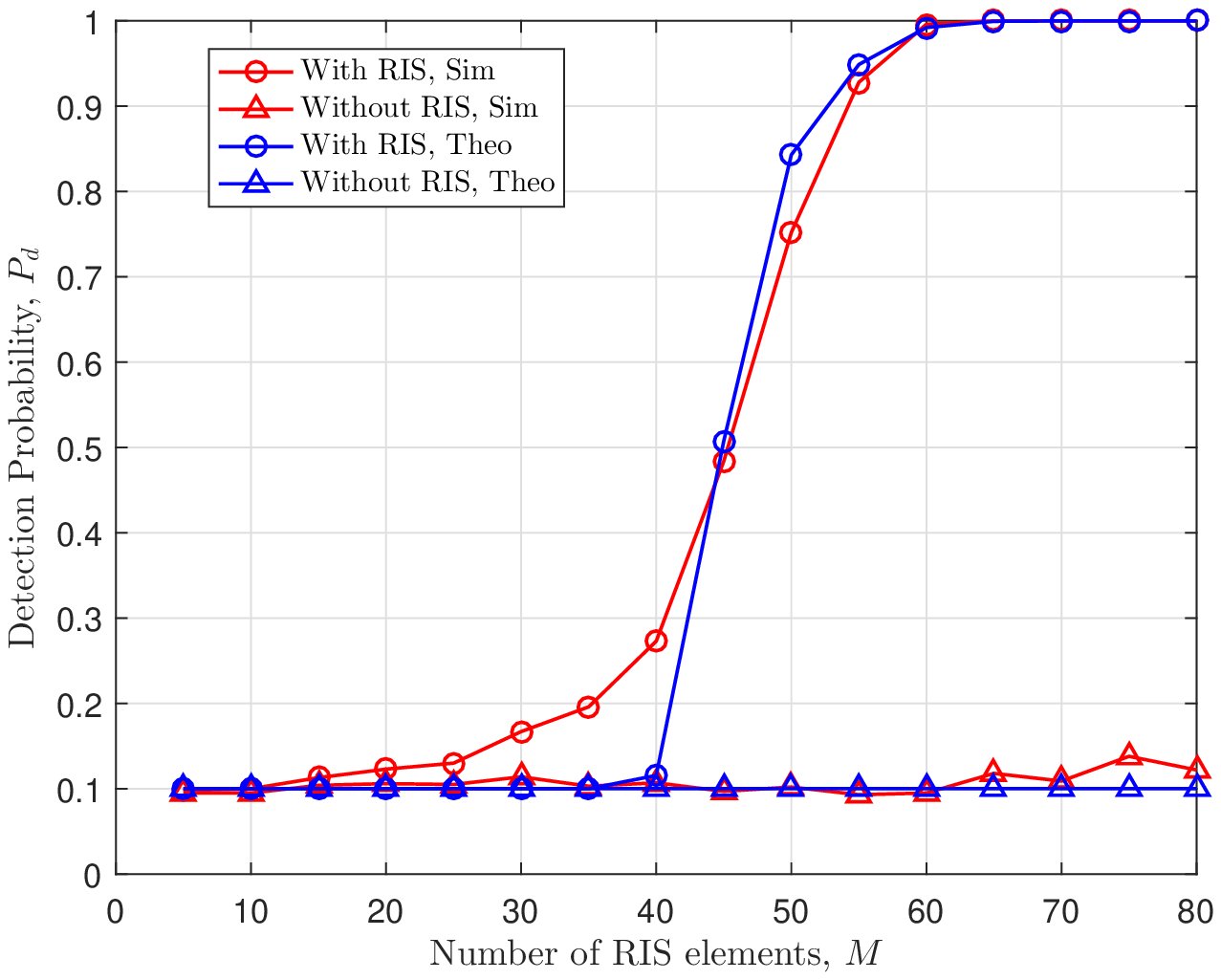}}
\subfigure[$N=128$]{%
\epsfxsize=0.24\textwidth \leavevmode
\epsffile{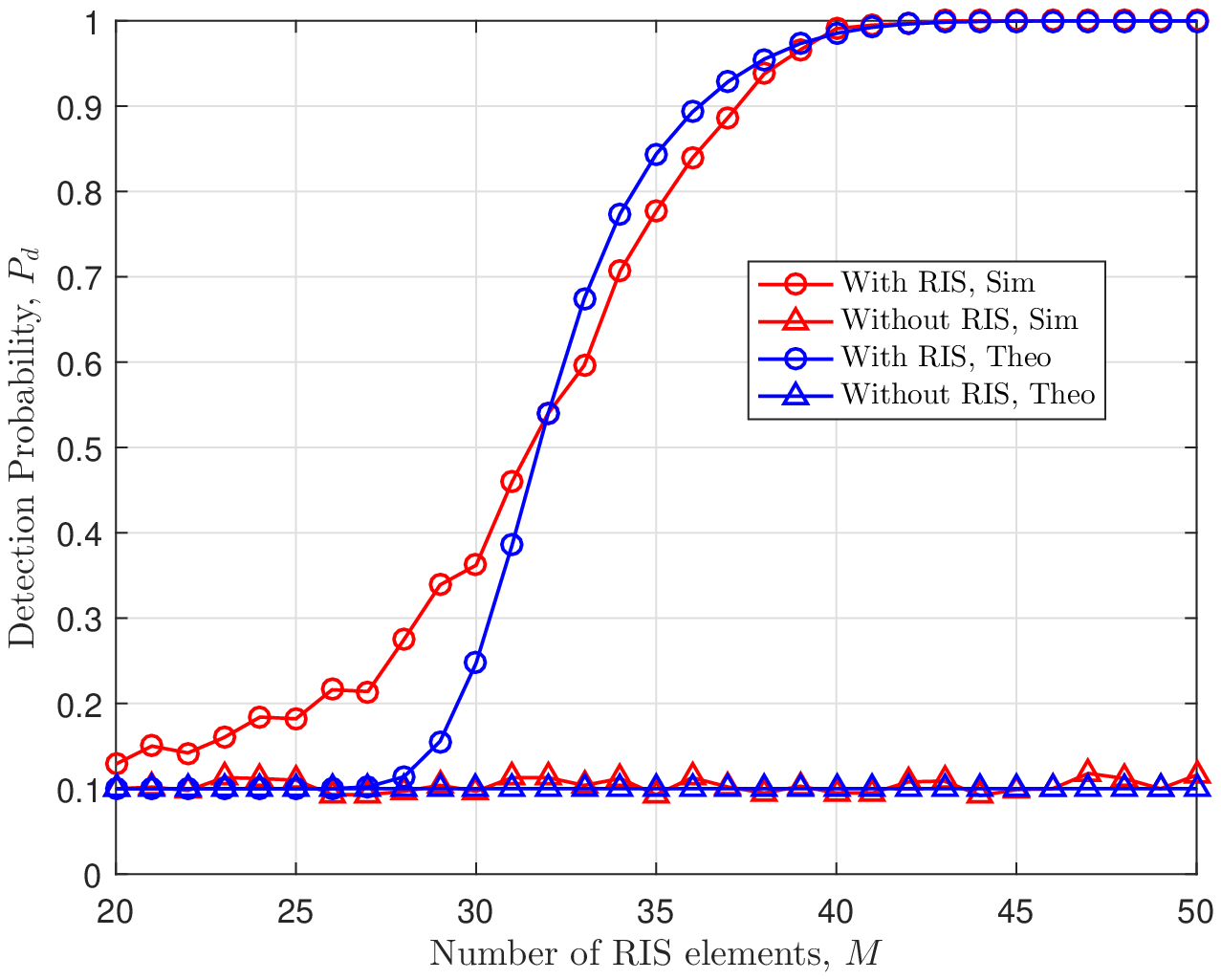}}
\subfigure[$N=256$]{%
\epsfxsize=0.24\textwidth \leavevmode
\epsffile{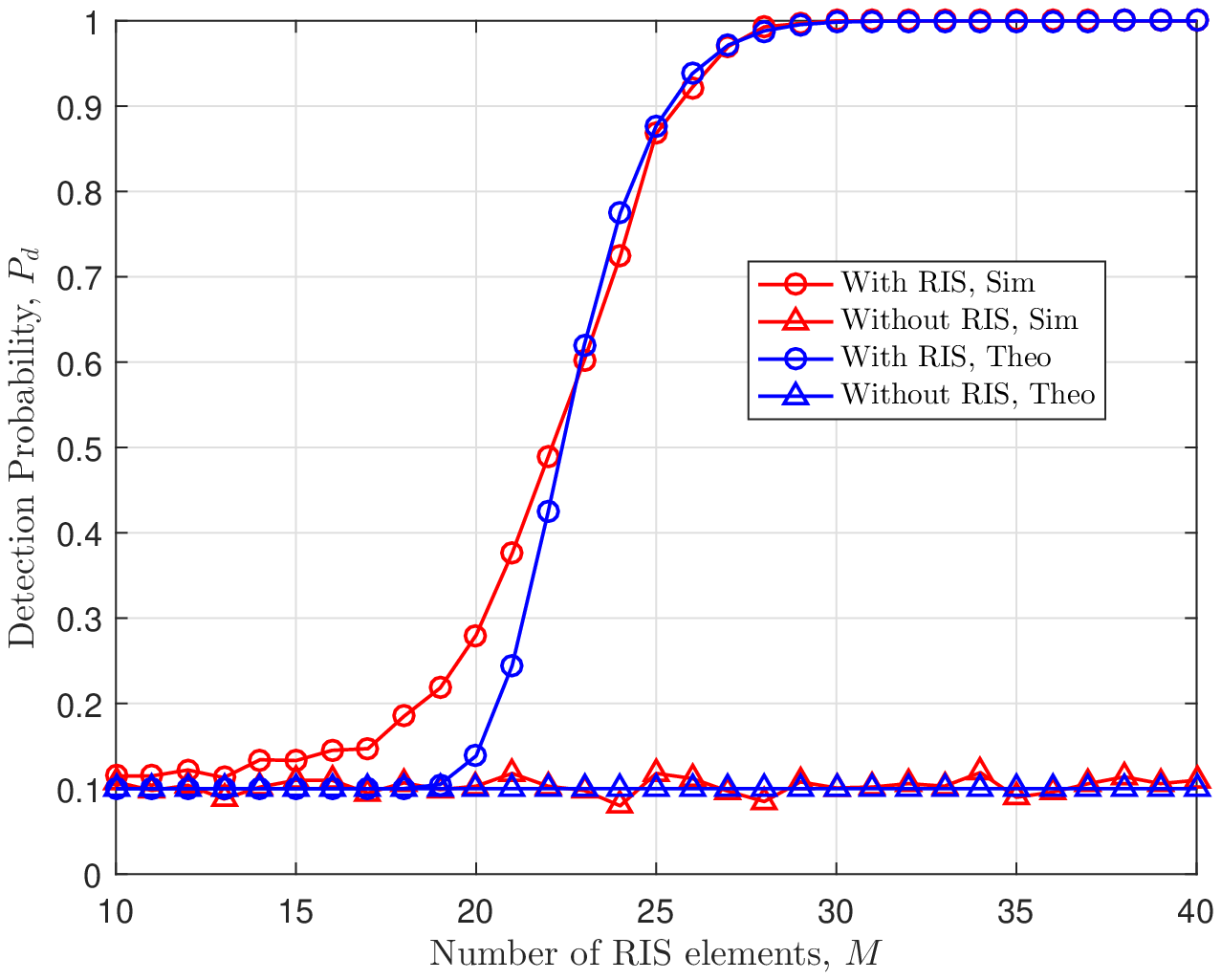}}
\subfigure[$N=512$]{%
\epsfxsize=0.24\textwidth \leavevmode
\epsffile{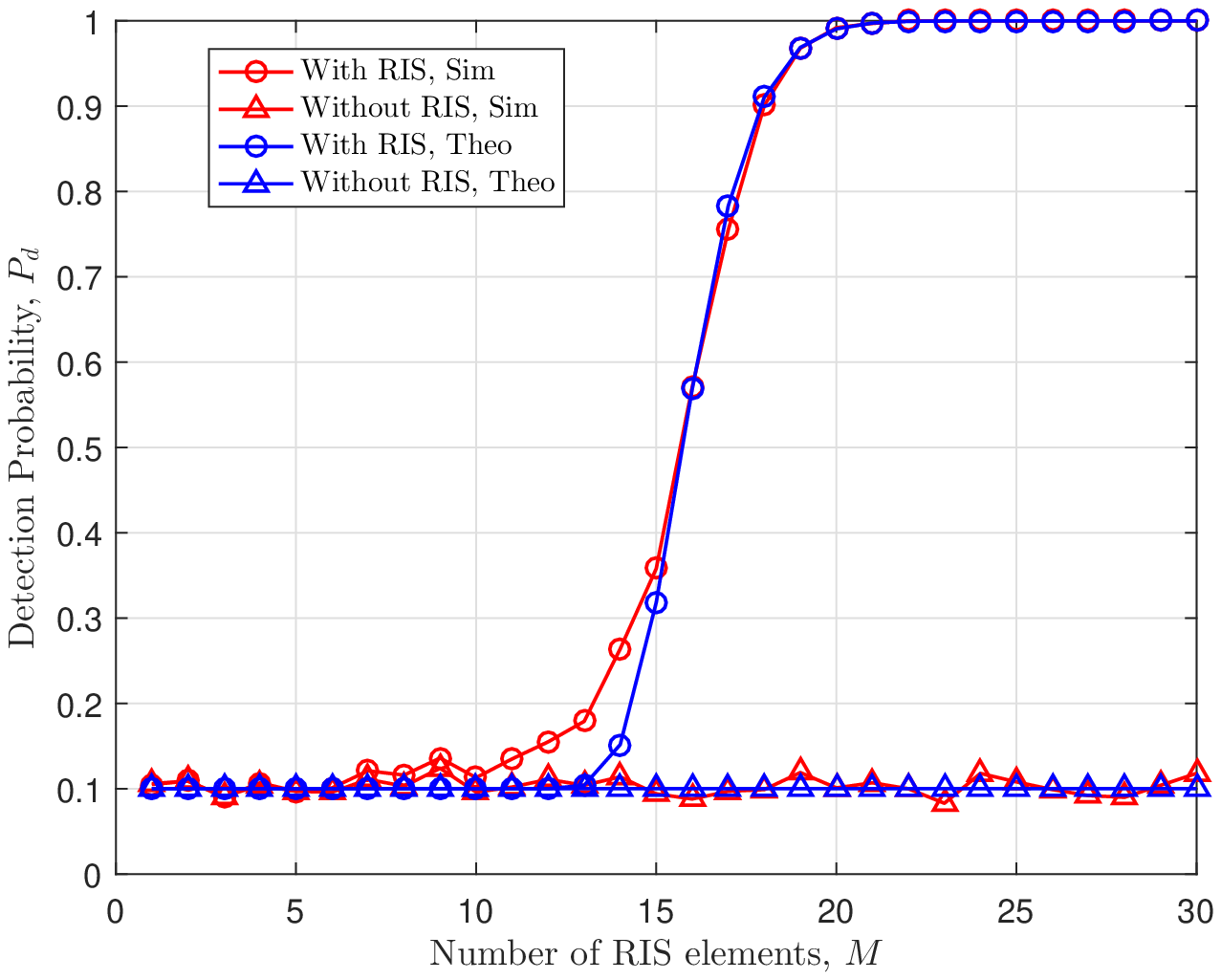}}
\caption{Comparison of the theoretical results and the Monte-Carlo simulation results of $P_d$ for different $M$'s and $N$'s.}\label{fig:RIS4SS_M_theo_sim}
\end{figure*}

We recall that the proposed analytical framework for evaluating $P_d$ is supposed to be accurate in the asymptotic regime. However, the number of antennas at ST can not be infinitely large even in the massive MIMO scenarios, where the number of antennas can be $64$, $128$ or even larger. With the fact that the Rayleigh case and the LoS case are actually special cases of the Rician RIS-related channels, we provide the results under the Rician RIS-related channels, where both $\kappa_f$ and $\kappa_G$ are set to $5$ without loss of generality. Besides, from the analysis in Section \ref{sec:RIS4SS}, the sensing performance mainly depends on the number of antennas at ST, i.e., $N$, and the number of signal samples, namely, $n=N/c$. To evaluate the impact of $N$ on the accuracy of the analytical framework, we fix $c=0.01$ in the simulations, i.e., $n=100N$. With these settings, the mean and variance of $g$ can be derived via \eqref{eq:mgRi}. Then, the analytical $P_d$ can be computed by numerically evaluating the double integral as shown in \eqref{eq:Pdint}. Besides, the Monte-Carlo simulation results are obtained via $1000$ random realizations for each pair of $M$ and $N$ under each hypothesis. The results for different $N$'s are shown in Fig. \ref{fig:RIS4SS_M_theo_sim}.

As shown in all the subfigures in Fig. \ref{fig:RIS4SS_M_theo_sim}, we have $P_d=P_{fa}$ for both the theoretical results and the simulation results when RIS is not employed. This phenomenon is consistent with our analysis for {\it Case} $1$ in Section \ref{sec:perfectdetection}. When the number of signal samples is not large enough to ensure $g>\sqrt{c}$ with a non-zero probability, we have $P_{d}\approx P_{fa}$. In the proposed RIS-enhanced spectrum sensing system, we can see that the detection probability increases as the number of REs increases. This is due to the fact that the RIS introduces an additional signal propagation path for improving the sensing SNR. Therefore, with the optimized phase shift matrix, $g>\sqrt{c}$ can almost surely holds by employing a large number of REs, which can be validated by the fact that the detection probability approaches $1$ when the number of REs becomes large. To conclude, with enough REs, the RIS-enhanced spectrum sensing system can achieve $P_{d}\approx 1$ without requiring more signal samples.

However, it can also be observed that the analytical results deviate a little from the Monte-Carlo simulation results. The reason is that Theorem \ref{thm:spike} actually holds in the asymptotic regime, where $N, n\to\infty$ with $N/n\to c$. When $N, n$ are not large enough, the asymptotic distributions of $T$ as shown in \eqref{eq:gdfleq} and \eqref{eq:gdfgeq} are shown to deviate a little from the true distributions \cite{couillet2011random}. Obviously, the simulation results with $N=64$ is far away from the asymptotic regime and therefore the gap between the analytical results and simulation results appears. This can be verified by the simulation results shown in Fig. \ref{fig:RIS4SS_M_theo_sim}(b), Fig. \ref{fig:RIS4SS_M_theo_sim}(c), and Fig. \ref{fig:RIS4SS_M_theo_sim}(d). As $N$ becomes larger, the accuracies of the analytical results are improved. Thus, we can make a conjecture that the analytical results are accurate in the asymptotic regime.

Despite the deviations of the analytical results in the finite dimensional regime, it can be observed that the analytical results approximate the simulation results well under where the detection probability approaches $1$. The reason is that, in the finite dimensional regime, the asymptotic distribution of $T$ under $\mathcal{H}_1$ is not so accurate when $g$ is near $\sqrt{c}$, i.e., \eqref{eq:gdfleq} is not accurate when $g$ is a little smaller than $\sqrt{c}$ and \eqref{eq:gdfgeq} is not accurate when $g$ is a little larger than $\sqrt{c}$. When $g$ is much larger than $\sqrt{c}$, \eqref{eq:gdfgeq} is still accurate for moderate $N$. Hence, we can infer that the theoretical predictions about the number of REs required to achieve $P_d\approx 1$ are quite accurate.

\subsection{Validity of the Theoretical Predictions about the Number of REs Required to Achieve $P_d\approx 1$}\label{subsec:Mtheoaccu}

In Section \ref{sec:perfectdetection}, we have provided the necessary condition and the sufficient condition of the number of REs required to achieve $P_d\approx 1$. We recall that $M_{\inf}$ is the lower bound and $M_{\rm PD}$ is large enough to achieve a detection probability close to $1$. Here, we provide the Monte-Carlo simulation results to evaluate the validity of $M_{\inf}$ and $M_{\rm PD}$. The number of antennas at ST, namely, $N$, is fixed at $64$. As aforementioned, more REs can ensure a higher $g$, and the RIS-enhanced spectrum sensing system can achieve $P_d\approx 1$ by letting $g>\sqrt{c}$ almost surely holds. As shown in Fig. \ref{fig:MpdMinf}, we plot the simulation results for the minimal number of REs required to achieve different levels of $P_d$, as $c$ increases. It is worth noting that when $N$ is fixed, larger $c$ means less signal samples for spectrum sensing. Besides, the theoretical $M_{\inf}$ and $M_{\rm PD}$ are also plotted to evaluate the validity of the theoretical predictions. Obviously, the simulation results are consistent with the theoretical predictions, i.e., $M_{\inf}$ is the lower bound of the number of REs required to achieve a high detection probability and $M_{\rm PD}$ is large enough to achieve $P_d\approx 1$. More specifically, $M_{\rm PD}$ can achieve a detection probability as high as $1-10^{-5}$. In addition, it can be observed that more REs are required to achieve a higher detection probability.

However, for spectrum sensing problems, it may be not necessary to realize a detection probability as high as $1-10^{-3}$ or a higher one. For example, in IEEE $802.22$ standard, achieving $P_d>0.9$ with $P_{fa}<0.1$ under $SNR=-20dB$ is enough to meet the requirements for TV signal detection \cite{zeng2010review}. Therefore, it is necessary to evaluate the accuracy of the theoretical predictions about the number of REs required to achieve a detection probability higher than $0.9$. In Fig. \ref{fig:cmptheosim}, the analytical results and the simulation results are provided. In particular, the theoretical predictions are obtained by searching the minimum $M$ which lets the integral in \eqref{eq:Pdint} satisfy the constraint of $P_d$, e.g., $P_d>0.9$. It can be observed that the analytical results approximate the simulation results well for $P_d>0.9$. In other words, the proposed asymptotic analytical framework for evaluating $P_{d}$ can actually provide us accurate predictions about the number of REs required to achieve a detection probability higher than $0.9$. Therefore, the theoretical predictions can help us determine the number of REs when we design the RIS-enhanced spectrum sensing system.

\begin{figure}[!t]
\centering
\epsfxsize=0.3\textwidth \leavevmode
\epsffile{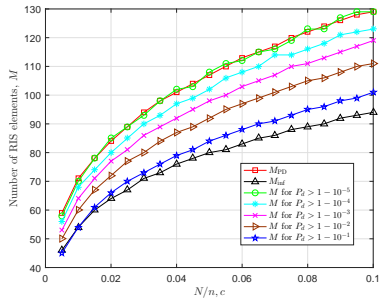}
\caption{The Monte-Carlo simulation results for the number of REs required to achieve different levels of $P_d$ and theoretical $M_{\inf}$, $M_{\rm PD}$.}\label{fig:MpdMinf}
\end{figure}

\begin{figure}[!t]
\centering
\epsfxsize=0.3\textwidth \leavevmode
\epsffile{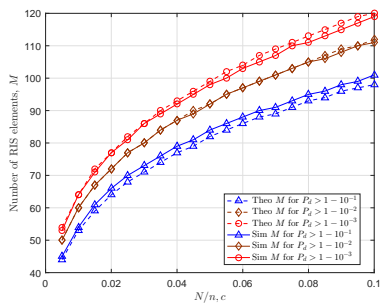}
\caption{Comparison between the analytical results and Monte-Carlo simulation results for the number of REs required to achieve different levels of $P_d$.}\label{fig:cmptheosim}
\end{figure}

\subsection{Impact of the Rician Factors of RIS-related Channels}\label{subsec:impactRicianFac}

\begin{figure*}[!t]
\centering
\subfigure[$\kappa_f=\kappa_G=0$]{%
\epsfxsize=0.25\textwidth \leavevmode
\epsffile{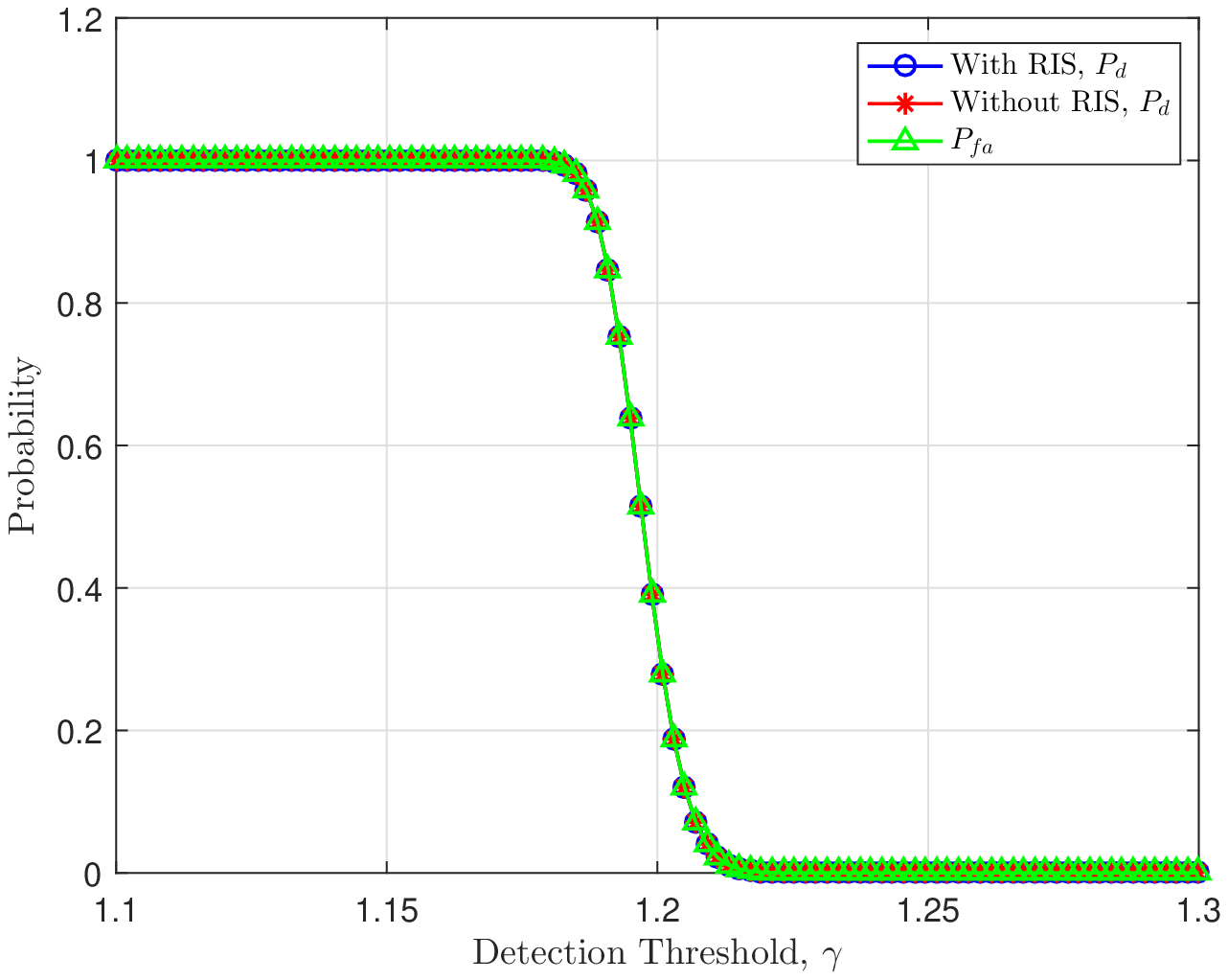}}\quad
\subfigure[$\kappa_f=\kappa_G=1$]{%
\epsfxsize=0.25\textwidth \leavevmode
\epsffile{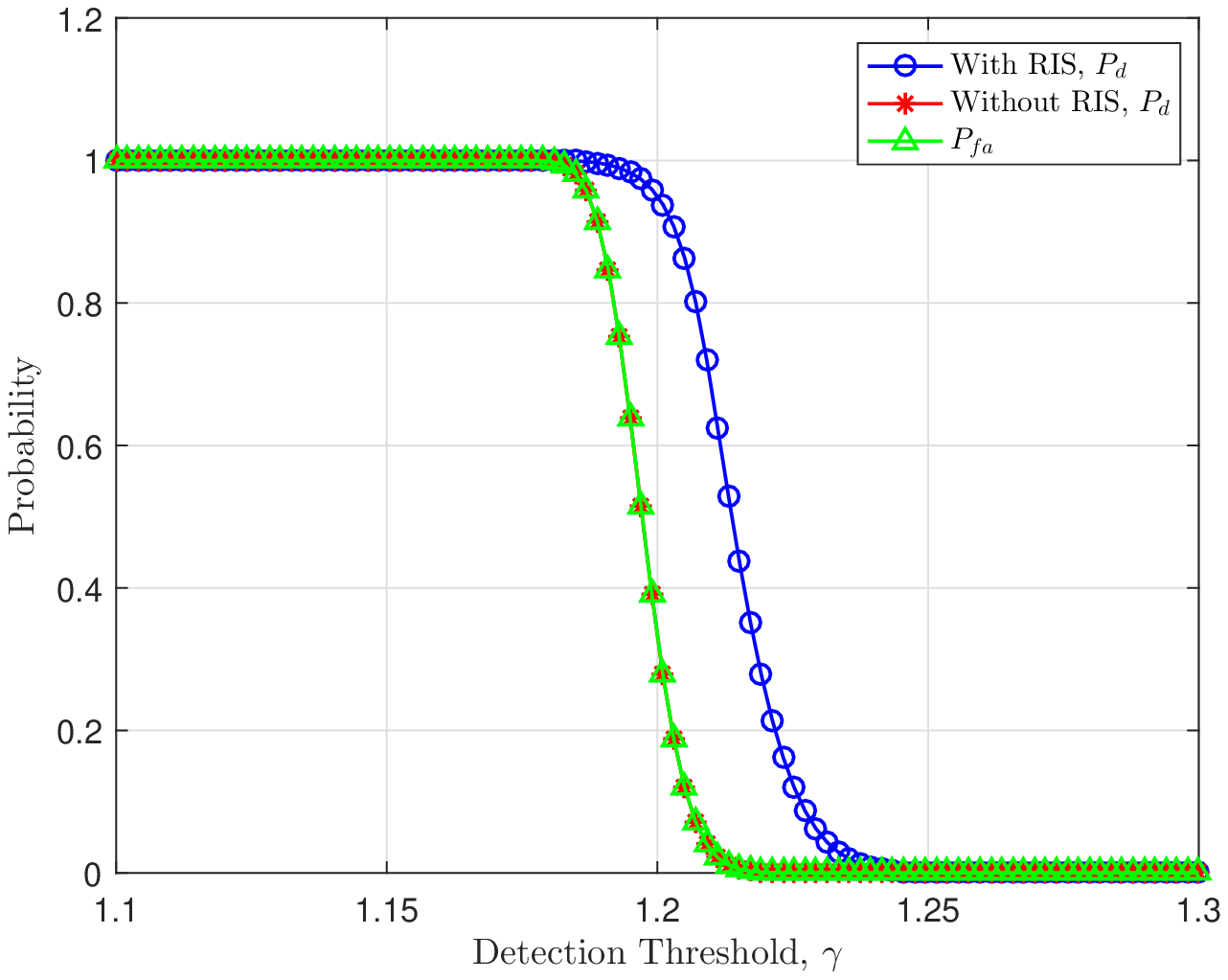}}\quad
\subfigure[$\kappa_f=\kappa_G=10$]{%
\epsfxsize=0.25\textwidth \leavevmode
\epsffile{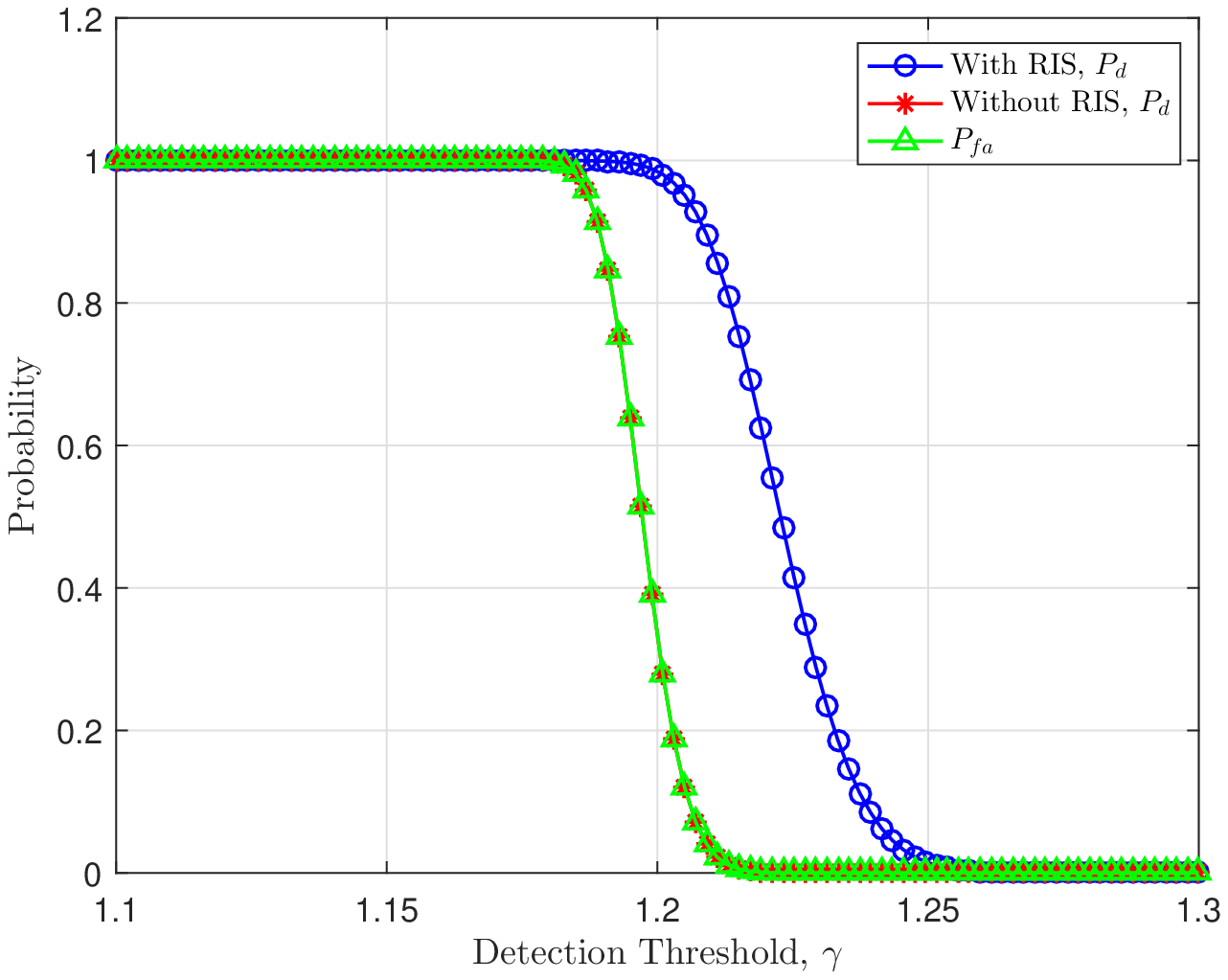}}\quad
\caption{Illustration for the number of REs required to achieve a detection probability close to $1$ under different RIS-related channel conditions.}\label{fig:illuchannelcond}
\end{figure*}

\begin{figure}[!t]
\centering
\epsfxsize=0.3\textwidth \leavevmode
\epsffile{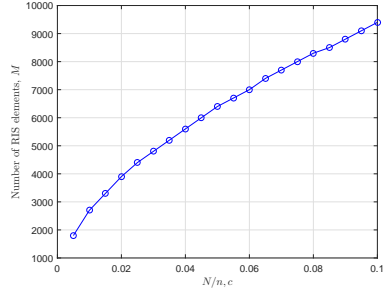}
\caption{Analytical results for the number of REs required to achieve $P_d>0.99$ under Rayleigh RIS-related channels, i.e., $\kappa_f=\kappa_G=0$.}\label{fig:NUMRay}
\end{figure}

\begin{figure}[!t]
\centering
\epsfxsize=0.3\textwidth \leavevmode
\epsffile{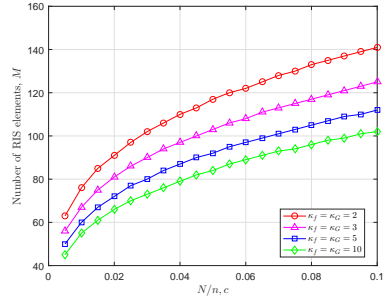}
\caption{Analytical results for the number of REs required to achieve $P_d>0.99$ under Rician RIS-related channels with different Rician factors.}\label{fig:NUMRi}
\end{figure}

Beyond the number of REs, it is obvious that the RIS-related channel Rician factors, namely, $\kappa_f$ and $\kappa_G$, also have a large impact on the detection probability. In this section, we mainly investigate the impact of $\kappa_f$ and $\kappa_G$ on the number of REs required to achieve $P_{d}>0.99$, since Fig. \ref{fig:cmptheosim} shows that the theoretical predictions about the number of REs required to achieve $P_d>0.99$ are rather accurate. Note that the LoS case and the Rayleigh case are actually special cases of Rician RIS-related channels, we first illustrate the impact of the Rician factors in Fig. \ref{fig:illuchannelcond}, which is obtained under where $N=64$, $M=50$ and $c=0.01$. In Fig. \ref{fig:illuchannelcond}, we plot the theoretical false alarm probability, and two theoretical detection probabilities, i.e., with RIS and without RIS, with respect to the detection threshold. As $\gamma$ increases, all the three probabilities decrease from $1$ to $0$. Moreover, due to the low SNR and the lack of signal samples, the distributions of the test statistic under the two hypotheses are the same, the curve of $P_d$ without RIS is thus overlapped with that of $P_{fa}$. Besides, under the Rayleigh case, i.e., $\kappa_f=\kappa_G=0$, the three curves are almost overlapped because of the small number of REs and the random phase shift design. On the contrary, under the Rician case, the RIS-enhanced spectrum sensing system can achieve a higher $P_{d}$ while the $P_{fa}$ is kept to be small. In addition, by comparing Fig. \ref{fig:illuchannelcond}(b) with \ref{fig:illuchannelcond}(c), it can be observed that, with the same number of REs, the gain of $P_{d}$ is more significant when the Rician factors are larger.

On the other hand, to achieve the same $P_{d}$, we can infer that less REs are required when the Rician factors are larger. Thus, we here also provide theoretical results for the number of REs required to achieve $P_{d}>0.99$ for different Rician factors when $N=64, P_{fa}=0.1$. When the RIS-related channels are of Rayleigh and the phase shift matrix is randomly chosen, the results are shown in Fig. \ref{fig:NUMRay}. To cover the large pathloss of the cascaded channel, thousands of REs are required for the Rayleigh case. Similarly, the results for the Rician case are shown in Fig. \ref{fig:NUMRi}. For the same $c$, i.e., the same number of signal samples, less REs are required to achieve the same $P_{d}$ as the channel Rician factors become larger. This phenomenon further justifies the fact that the statistical phase shift design approaches the optimal as the channel Rician factors go to infinity.

\section{Conclusion}\label{sec:conclu}

In this study, we have proposed an RIS-enhanced spectrum sensing system to improve the detection performance without a large number of sensing signal samples. To this end, we have firstly proposed an asymptotic analytical framework for evaluating the detection probability of the MED approach. Then, we have presented the statistical design for the phase shift matrix of the RIS, which only requires the statistical CSI. With this design, the asymptotic distributions of the equivalent channel gains have been derived. Combined with the analytical framework for evaluating $P_{d}$, we have analyzed the necessary condition and the sufficient condition of the number of REs required to achieve a detection probability close to $1$, respectively. Last but not least, the Monte-Carlo simulation results have been provided to evaluate the validity of the theoretical preditions, and the results have shown that the proposed RIS-enhanced can substantially improve the detection probability without more sensing signal samples.

%\section*{Acknowledgment}
%This research was supported in part by National Natural Science Foundation of China under Grants 61631005 and U1801261 and in part by the National Research Foundation of Korea (NRF) grants funded by the Korea government (MSIT) (NRF-2019R1A2C1084168).

\appendix

\subsection{Proof of Lemma \ref{lem:rdfRay}}\label{apdx:prooflem1}
\begin{prop}\label{prop:tracelemma}
Consider a series of symmetric matrices $\mathbf{A}_1, \mathbf{A}_2, \cdots$, with $\mathbf{A}_M\in\mathbb{C}^{M\times M}$ which have uniformly bounded spectrum norm and a series of vectors $\mathbf{x}_1, \mathbf{x}_2, \cdots,$ with $\mathbf{x}_M\in\mathbb{C}^{M}$ which have {\it i.i.d.} elements of zero mean, variance $1/M$ and eighth moment of order $O(1/M^4)$, independent of $\mathbf{A}_M$. Then
\begin{equation}\label{eq:tracelemmamean}
\mathbf{x}_{M}^H\mathbf{A}_M\mathbf{x}_{M}-\frac{1}{M}\Tr(\mathbf{A}_M)\stackrel{a.s.}{\longrightarrow}0
\end{equation}
Besides, for the real case where the involved quantities are real and the entries of $\mathbf{x}_M$ have fourth order moment of order $O(1/M^2)$, assume $\mathbf{A}_M$ has a limit spectrum density, $F^{\mathbf{A}}$, \cite{tse2000linear} gives a CLT for $\mathbf{x}_{M}^H\mathbf{A}_M\mathbf{x}_{M}-\frac{1}{M}tr\mathbf{A}_M$, i.e.,
\begin{equation}\label{eq:tracelemmaCLT}
\sqrt{M}[\mathbf{x}_{M}^H\mathbf{A}_M\mathbf{x}_{M}-\frac{1}{M}\Tr(\mathbf{A}_M)]\xrightarrow[M\to\infty]{\mathcal{D}}Z\sim\mathcal{N}(0, v),
\end{equation}
where the variance $v$ depends on the limit spectrum density of $\mathbf{A}_M$, namely, $F^{\mathbf{A}}$, and the fourth moment of the entries of $\mathbf{x}_M$ as
\begin{equation}\label{eq:tracelemmavarreal}
v = 2\int t^2\dif F^{\mathbf{A}}(t)+(\mathbb{E}[x_{11}^4]-3)\left(\int t\dif F^{\mathbf{A}}(t)\right)^2,
\end{equation}
and $x_{11}$ satisfies the same distribution as the entries of $\mathbf{x}_M$ but with unit variance.

However, \cite{tse2000linear} only provides the results under the real case. When $x_{11}$ is complex with $\mathbb{E}x_{11}^2=0$, we will show that the variance becomes
\begin{equation}\label{eq:tracelemmavarcmp}
v = \int t^2\dif f^{\mathbf{A}}(t)+(\mathbb{E}[|x_{11}|^4]-2)\left(\int t\dif F^{\mathbf{A}}(t)\right)^2.
\end{equation}
\end{prop}
\begin{proof}
It can be verified that \cite{liang2007asymptotic, fang2014spectral}
\begin{align}\label{eq:tracelemmaCLTorigin}
  v&=\mathbb{E}\left[\mathbf{x}_M^H\mathbf{A}_M\mathbf{x}_M-\frac{1}{M}\Tr(\mathbf{A}_M)\right]^2 \nonumber\\
  &=\frac{1}{M}\left[\Tr(\mathbf{A}_M^2)+|\mathbb{E}(x_{11}^2)|^2\Tr(\mathbf{A}_M\mathbf{A}_M^H)\right]+\frac{1}{M}\left[\sum_{i=1}^{M}A_{M, ii}^2(\mathbb{E}|x_{11}|^4-2-|\mathbb{E}x_{11}^2|^2)\right].
\end{align}
Since $\mathbf{A}_M$ is Hermitian, therefore $\Tr(\mathbf{A}_M^2)=\Tr(\mathbf{A}_M\mathbf{A}_M^H)$ and the diagonal elements of $\mathbf{A}_M$ are real. In addition, with the law of large number (LLN), we have
\begin{equation}\label{eq:meantrA2}
\lim_{M\to\infty}\frac{1}{M}\Tr(\mathbf{A}_M^2)=\lim_{M\to\infty}\frac{1}{M}\sum_{i=1}^{M}\lambda_{i}^2=\int t^2\dif F^{\mathbf{A}}(t).
\end{equation}
On the other hand, using \cite[Theorem $1$]{bai2007asymptotics}, we can obtain
\begin{equation}\label{eq:limitAii}
\lim_{M\to\infty}\left|A_{M, ii}-\int t\dif F^{\mathbf{A}}(t)\right|\xrightarrow{a.s.} 0.
\end{equation}
Therefore, for the real case with $\mathbb{E}x_{11}^2=1$ and the complex case with $\mathbb{E}x_{11}^2=0$, \eqref{eq:tracelemmavarreal} and \eqref{eq:tracelemmavarcmp} can be easily obtained by substituting \eqref{eq:meantrA2} and \eqref{eq:limitAii} into \eqref{eq:tracelemmaCLTorigin}.
\end{proof}

Let $\tilde{\mathbf{f}}\in\mathbb{C}^{M}$ be a complex Gaussian random vector with zero mean and covariance matrix $\mathbf{I}_M$, $\tilde{\mathbf{G}}\in\mathbb{C}$ be a matrix whose entries are {\it i.i.d.} complex Gaussian random variables with zero mean and unit variance, $\mathbf{\Phi}$ be an arbitrary complex diagonal matrix whose diagonal elements are of unit modulus. To obtain the distribution of $\tilde{\mathbf{f}}^H\mathbf{\Phi}^H\tilde{\mathbf{G}}^H\tilde{\mathbf{G}}\mathbf{\Phi}\tilde{\mathbf{f}}$, we first give a scaled version of \eqref{eq:tracelemmaCLT}, i.e.,
\begin{equation}\label{eq:tracelemmaCLTscaled}
\tilde{\mathbf{f}}^H\mathbf{A}_M\tilde{\mathbf{f}}-\Tr(\mathbf{A}_M)\xrightarrow[M\to\infty]{\mathcal{D}}Z\sim\mathcal{N}(0, Mv),
\end{equation}
where $\mathbf{A}_M$, $v$ are defined as aforementioned. \eqref{eq:tracelemmaCLTscaled} can be easily proved with the fact that $\mathbb{E}[\tilde{\mathbf{f}}^H\mathbf{A}_M\tilde{\mathbf{f}}]=M\mathbb{E}[\mathbf{x}_{M}^H\mathbf{A}_M\mathbf{x}_{M}]$.
Then, the asymptotic distribution of $\tilde{\mathbf{f}}^H\mathbf{\Phi}^H\tilde{\mathbf{G}}^H\tilde{\mathbf{G}}\mathbf{\Phi}\tilde{\mathbf{f}}$ can be directly derived by letting $\mathbf{A}_M=\mathbf{\Phi}^H\tilde{\mathbf{G}}^H\tilde{\mathbf{G}}\mathbf{\Phi}$ in \eqref{eq:tracelemmaCLTscaled}. The expectation of $\tilde{\mathbf{f}}^H\mathbf{A}_M\tilde{\mathbf{f}}$ is exactly $\mathbb{E}\Tr[\mathbf{\Phi}^H\tilde{\mathbf{G}}^H\tilde{\mathbf{G}}\mathbf{\Phi}]$ and the variance can be accurately approximated by $Mv$ \cite{tse2000linear}, while $v$ depends on the limit spectrum density of $\mathbf{A}_M$, namely, $F^{\mathbf{A}}$, and the fourth moment of the entries of $\tilde{\mathbf{f}}$. Denoting the $i$-th column of $\mathbf{G}^H$ by $\tilde{\mathbf{g}}_i$, the expectation value can be calculated as
\begin{align}\label{eq:traceAm}
\mathbb{E}\Tr(\mathbf{\Phi}^H\tilde{\mathbf{G}}^H\tilde{\mathbf{G}}\mathbf{\Phi})&=\mathbb{E}\Tr(\tilde{\mathbf{G}}^H\tilde{\mathbf{G}}) =\Tr\left(\sum_{i=1}^{N}\mathbb{E}[\tilde{\mathbf{g}}_i\tilde{\mathbf{g}}_i^H]\right)
=\Tr(N\mathbf{I}_M)=MN.
\end{align}
Note that both $\tilde{\mathbf{f}}$ and $\tilde{\mathbf{G}}$ are complex, the parameter, $v$, in \eqref{eq:tracelemmaCLTscaled} can be computed by \eqref{eq:tracelemmavarcmp}. Specifically, the first order moment of $F^{\mathbf{A}}$ can be derived as
\begin{align}\label{eq:mpm1}
  \int t\dif F^{\mathbf{A}}(t)&=\lim_{M\to\infty}\frac{1}{M}\mathbb{E}\Tr(\mathbf{A}_M)
  =\lim_{M\to\infty}\frac{1}{M}\mathbb{E}\Tr(\mathbf{\Phi}^H\tilde{\mathbf{G}}^H\tilde{\mathbf{G}}\mathbf{\Phi})
  =\frac{1}{M}MN=N.
\end{align}
Besides, the second order moment of $F^{\mathbf{A}}$ can be decomposed as
\begin{align}\label{eq:mpm2}
  \int t^2\dif f^{\mathbf{A}}(t)&=\lim_{M\to\infty}\frac{1}{M}\mathbb{E}\Tr(A^2)
  =\lim_{M\to\infty}\frac{1}{M}\mathbb{E}\Tr(\mathbf{\Phi}^H\tilde{\mathbf{G}}^H\tilde{\mathbf{G}}\mathbf{\Phi}\mathbf{\Phi}^H\tilde{\mathbf{G}}^H\tilde{\mathbf{G}}\mathbf{\Phi})\nonumber\\
  &=\lim_{M\to\infty}\frac{1}{M}\mathbb{E}\Tr(\tilde{\mathbf{G}}^H\tilde{\mathbf{G}}\tilde{\mathbf{G}}^H\tilde{\mathbf{G}})
  =\lim_{M\to\infty}\frac{1}{M}\Tr\left[\mathbb{E}\left(\sum_{i=1}^{N}\tilde{\mathbf{g}}_i\tilde{\mathbf{g}}_i^H\sum_{j=1}^{N}\tilde{\mathbf{g}}_j\tilde{\mathbf{g}}_j^H\right)\right]\nonumber\\
  &=\lim_{M\to\infty}\frac{1}{M}\Tr\left[\mathbb{E}\left(\sum_{i=1}^{N}\tilde{\mathbf{g}}_i\tilde{\mathbf{g}}_i^H\sum_{j=1,j\neq i}^{N}\tilde{\mathbf{g}}_j\tilde{\mathbf{g}}_j^H+\sum_{i=1}^{N}\tilde{\mathbf{g}}_i\tilde{\mathbf{g}}_i^H\tilde{\mathbf{g}}_i\tilde{\mathbf{g}}_i^H\right)\right]\nonumber\\
  &=\lim_{M\to\infty}\frac{1}{M}\left[\Tr\left(\sum_{i=1}^{N}\mathbb{E}[\tilde{\mathbf{g}}_i\tilde{\mathbf{g}}_i^H]\sum_{j=1,j\neq i}^{N}\mathbb{E}[\tilde{\mathbf{g}}_j\tilde{\mathbf{g}}_j^H]\right)+ \sum_{i=1}^{N}
  \Tr\left(\mathbb{E}[\tilde{\mathbf{g}}_i\tilde{\mathbf{g}}_i^H\tilde{\mathbf{g}}_i\tilde{\mathbf{g}}_i^H]\right)\right].
\end{align}
The first component of \eqref{eq:mpm2} can be simply derived by
\begin{align}\label{eq:mpm2p1}
\Tr\left(\sum_{i=1}^{N}\mathbb{E}[\tilde{\mathbf{g}}_i\tilde{\mathbf{g}}_i^H]\sum_{j=1,j\neq i}^{N}\mathbb{E}[\tilde{\mathbf{g}}_j\tilde{\mathbf{g}}_j^H]\right)&=\Tr[N(N-1)\mathbf{I}_M]
=(N-1)MN.
\end{align}
Denoting $g_{ip}$ the $p$-th element of $\mathbf{g}_i$, the second component of \eqref{eq:mpm2} becomes
\begin{align}\label{eq:mpm2p2}
  \mathbb{E}\Tr\left(\tilde{\mathbf{g}}_i\tilde{\mathbf{g}}_i^H\tilde{\mathbf{g}}_i\tilde{\mathbf{g}}_i^H\right)&=\mathbb{E}\Tr\left(\tilde{\mathbf{g}}_i^H\tilde{\mathbf{g}}_i\tilde{\mathbf{g}}_i^H\tilde{\mathbf{g}}_i\right)=
    \mathbb{E}\tilde{\mathbf{g}}_i^H\tilde{\mathbf{g}}_i\tilde{\mathbf{g}}_i^H\tilde{\mathbf{g}}_i
  =\mathbb{E}\left[\left(\sum_{p=1}^{M}|g_{ip}|^2\right)\left(\sum_{q=1}^{M}|g_{iq}|^2\right)\right]\nonumber\\
  &=\sum_{p=1}^{M}\sum_{q\neq p}^{M}\mathbb{E}[|g_{ip}|^2]\mathbb{E}[|g_{iq}|^2]+\sum_{p=1}^{M}\mathbb{E}[|g_{ip}|^4]
  =M(M-1)+M\mathbb{E}[|g_{ip}|^4].
\end{align}
Note that $g_{ip}$ is a complex Gaussian random variable with zero mean and unit variance, we have $\mathbb{E}[|g_{ip}|^4]=2$. Therefore, \eqref{eq:mpm2} becomes
\begin{align}\label{eq:mpm2all}
  \int t^2\dif f^{\mathbf{A}}(t)&=\lim_{M\to\infty}\frac{1}{M}[(N-1)MN+NM(M+1)]
  =(M+N)N.
\end{align}
In addition, the entries of $\tilde{\mathbf{f}}$ are also {\it i.i.d} complex Gaussian random variables with zero mean and unit variance, i.e., $\mathbb{E}[|x_{11}|^4]=2$, and \eqref{eq:tracelemmavarcmp} becomes
\begin{equation}\label{eq:vcal}
v=\int t^2\dif f^{\mathbf{A}}(t)=(M+N)N.
\end{equation}

Substituting \eqref{eq:traceAm} and \eqref{eq:vcal} into \eqref{eq:tracelemmaCLTscaled}, the mean and variance of $\tilde{\mathbf{f}}^H\mathbf{\Phi}^H\tilde{\mathbf{G}}^H\tilde{\mathbf{G}}\mathbf{\Phi}\tilde{\mathbf{f}}$ are $MN$ and $(M+N)MN$, respectively. Note that $\mathbf{f}=\sqrt{\beta}_f\tilde{\mathbf{f}}$ and $\mathbf{G}=\sqrt{\beta}_G\tilde{\mathbf{G}}$, the mean and variance of $r=\|\mathbf{G}\mathbf{\Phi}\mathbf{f}\|^2$ are therefore
\begin{equation}\label{eq:mrRayproof}
\mu_r = MN\beta_f\beta_G,\  v_r = (M+N)MN\beta_f^2\beta_G^2.
\end{equation}
%\begin{equation}\label{eq:vrRayproof}
%v_r = (M+N)MN\beta_f^2\beta_G^2.
%\end{equation}

\subsection{Proof of Lemma \ref{lem:rdfRi}}\label{apdx:prooflem2}
For ease of notation, let $r'=\|\mathbf{G}'\mathbf{\Phi}\mathbf{f}'\|^2$, we here consider a simpler case where the coefficients, namely, $\bar{\beta}_f$, $\tilde{\beta}_f$, $\bar{\beta}_G$, $\tilde{\beta}_G$ defined in \eqref{eq:betafbt} are ignored, i.e, $\mathbf{f}'=\bar{\mathbf{f}}+\tilde{\mathbf{f}}$, $\mathbf{G}'=\bar{\mathbf{G}}+\tilde{\mathbf{G}}$. As a consequence, $r'$ becomes
\begin{align}
r'=\|(\bar{\mathbf{G}}+\tilde{\mathbf{G}})\mathbf{\Phi}(\bar{\mathbf{f}}+\tilde{\mathbf{f}})\|^2
=\|\bar{\mathbf{G}}\mathbf{\Phi}\bar{\mathbf{f}}+\bar{\mathbf{G}}\mathbf{\Phi}\tilde{\mathbf{f}}+\tilde{\mathbf{G}}\mathbf{\Phi}\bar{\mathbf{f}}
+\tilde{\mathbf{G}}\mathbf{\Phi}\tilde{\mathbf{f}}\|^2.\label{eq:r_simple}
\end{align}
Obviously, the expansion of \eqref{eq:r_simple} has $16$ components as
\begin{align}\label{eq:r_simple_exp}
r&=\bar{\mathbf{f}}^H\mathbf{\Phi}^H\bar{\mathbf{G}}^H\bar{\mathbf{G}}\mathbf{\Phi}\bar{\mathbf{f}}
+\bar{\mathbf{f}}^H\mathbf{\Phi}^H\bar{\mathbf{G}}^H\bar{\mathbf{G}}\mathbf{\Phi}\tilde{\mathbf{f}}
+\bar{\mathbf{f}}^H\mathbf{\Phi}^H\bar{\mathbf{G}}^H\tilde{\mathbf{G}}\mathbf{\Phi}\bar{\mathbf{f}}
+\bar{\mathbf{f}}^H\mathbf{\Phi}^H\bar{\mathbf{G}}^H\tilde{\mathbf{G}}\mathbf{\Phi}\tilde{\mathbf{f}}\nonumber\\
&\quad+\tilde{\mathbf{f}}^H\mathbf{\Phi}^H\bar{\mathbf{G}}^H\bar{\mathbf{G}}\mathbf{\Phi}\bar{\mathbf{f}}
+\tilde{\mathbf{f}}^H\mathbf{\Phi}^H\bar{\mathbf{G}}^H\bar{\mathbf{G}}\mathbf{\Phi}\tilde{\mathbf{f}}
+\tilde{\mathbf{f}}^H\mathbf{\Phi}^H\bar{\mathbf{G}}^H\tilde{\mathbf{G}}\mathbf{\Phi}\bar{\mathbf{f}}
+\tilde{\mathbf{f}}^H\mathbf{\Phi}^H\bar{\mathbf{G}}^H\tilde{\mathbf{G}}\mathbf{\Phi}\tilde{\mathbf{f}}\nonumber\\
&\quad+\bar{\mathbf{f}}^H\mathbf{\Phi}^H\tilde{\mathbf{G}}^H\bar{\mathbf{G}}\mathbf{\Phi}\bar{\mathbf{f}}
+\bar{\mathbf{f}}^H\mathbf{\Phi}^H\tilde{\mathbf{G}}^H\bar{\mathbf{G}}\mathbf{\Phi}\tilde{\mathbf{f}}
+\bar{\mathbf{f}}^H\mathbf{\Phi}^H\tilde{\mathbf{G}}^H\tilde{\mathbf{G}}\mathbf{\Phi}\bar{\mathbf{f}}
+\bar{\mathbf{f}}^H\mathbf{\Phi}^H\tilde{\mathbf{G}}^H\tilde{\mathbf{G}}\mathbf{\Phi}\tilde{\mathbf{f}}\nonumber\\
&\quad+\tilde{\mathbf{f}}^H\mathbf{\Phi}^H\tilde{\mathbf{G}}^H\bar{\mathbf{G}}\mathbf{\Phi}\bar{\mathbf{f}}
+\tilde{\mathbf{f}}^H\mathbf{\Phi}^H\tilde{\mathbf{G}}^H\bar{\mathbf{G}}\mathbf{\Phi}\tilde{\mathbf{f}}
+\tilde{\mathbf{f}}^H\mathbf{\Phi}^H\tilde{\mathbf{G}}^H\tilde{\mathbf{G}}\mathbf{\Phi}\bar{\mathbf{f}}
+\tilde{\mathbf{f}}^H\mathbf{\Phi}^H\tilde{\mathbf{G}}^H\tilde{\mathbf{G}}\mathbf{\Phi}\tilde{\mathbf{f}}.
\end{align}
According to the CLT, we can assume that $r'$ satisfies a Gaussian distribution with mean $\mu_{r'}$ and variance $v_{r'}$. The two parameters can then be calculated with
\begin{align}
  \mu_{r'} = \mathbb{E}[r'],
  v_{r'} = \mathbb{E}[{r'}^2]-\mathbb{E}[r']^2.
\end{align}
Taking the expectation of \eqref{eq:r_simple_exp}, it can be observed that the expectations of $12$ components of \eqref{eq:r_simple_exp} are $0$ since they have unpaired $\tilde{\mathbf{f}}$ or $\tilde{\mathbf{G}}$. Therefore, the remaining non-zero four components are
\begin{align}\label{eq:mrproof}
\mu_{r'}&=\mathbb{E}[\bar{\mathbf{f}}^H\mathbf{\Phi}^H\bar{\mathbf{G}}^H\bar{\mathbf{G}}\mathbf{\Phi}\bar{\mathbf{f}}]
+\mathbb{E}[\tilde{\mathbf{f}}^H\mathbf{\Phi}^H\bar{\mathbf{G}}^H\bar{\mathbf{G}}\mathbf{\Phi}\tilde{\mathbf{f}}]
+\mathbb{E}[\bar{\mathbf{f}}^H\mathbf{\Phi}^H\tilde{\mathbf{G}}^H\tilde{\mathbf{G}}\mathbf{\Phi}\bar{\mathbf{f}}]
+\mathbb{E}[\tilde{\mathbf{f}}^H\mathbf{\Phi}^H\tilde{\mathbf{G}}^H\tilde{\mathbf{G}}\mathbf{\Phi}\tilde{\mathbf{f}}].
\end{align}
Obviously, with the statistical phase shift matrix design in \eqref{eq:optimalPSMLoS}, $\bar{\mathbf{f}}^H\mathbf{\Phi}^H\bar{\mathbf{G}}^H\bar{\mathbf{G}}\mathbf{\Phi}\bar{\mathbf{f}}$ is a constant, i.e.,
\begin{equation}\label{eq:Ebbbb}
\mathbb{E}[\bar{\mathbf{f}}^H\mathbf{\Phi}^H\bar{\mathbf{G}}^H\bar{\mathbf{G}}\mathbf{\Phi}\bar{\mathbf{f}}]=
\bar{\mathbf{f}}^H\mathbf{\Phi}^H\bar{\mathbf{G}}^H\bar{\mathbf{G}}\mathbf{\Phi}\bar{\mathbf{f}}=M^2N.
\end{equation}
With Proposition \ref{prop:tracelemma} and the fact that $$\mathbb{E}\Tr(\tilde{\mathbf{G}}^H\tilde{\mathbf{G}})=\mathbb{E}\Tr(\bar{\mathbf{G}}^H\bar{\mathbf{G}})=MN,$$
$\mathbb{E}[\tilde{\mathbf{f}}^H\mathbf{\Phi}^H\bar{\mathbf{G}}^H\bar{\mathbf{G}}\mathbf{\Phi}\tilde{\mathbf{f}}]$,
$\mathbb{E}[\tilde{\mathbf{f}}^H\mathbf{\Phi}^H\tilde{\mathbf{G}}^H\tilde{\mathbf{G}}\mathbf{\Phi}\tilde{\mathbf{f}}]$ can be calculated by
\begin{align}\label{eq:Etbbt}
  \mathbb{E}[\tilde{\mathbf{f}}^H\mathbf{\Phi}^H\bar{\mathbf{G}}^H\bar{\mathbf{G}}\mathbf{\Phi}\tilde{\mathbf{f}}] &= \mathbb{E}\Tr[\mathbf{\Phi}^H\bar{\mathbf{G}}^H\bar{\mathbf{G}}\mathbf{\Phi}]
  =\mathbb{E}\Tr[\bar{\mathbf{G}}^H\bar{\mathbf{G}}]=MN,\\
  \mathbb{E}[\tilde{\mathbf{f}}^H\mathbf{\Phi}^H\tilde{\mathbf{G}}^H\tilde{\mathbf{G}}\mathbf{\Phi}\tilde{\mathbf{f}}] &= \mathbb{E}\Tr[\mathbf{\Phi}^H\tilde{\mathbf{G}}^H\tilde{\mathbf{G}}\mathbf{\Phi}]
  =\mathbb{E}\Tr[\tilde{\mathbf{G}}^H\tilde{\mathbf{G}}]=MN.
\end{align}
%\begin{align}\label{eq:Ebttb}
%  \mathbb{E}[\tilde{\mathbf{f}}^H\mathbf{\Phi}^H\tilde{\mathbf{G}}^H\tilde{\mathbf{G}}\mathbf{\Phi}\tilde{\mathbf{f}}] &= \mathbb{E}\Tr[\mathbf{\Phi}^H\tilde{\mathbf{G}}^H\tilde{\mathbf{G}}\mathbf{\Phi}]
%  =\mathbb{E}\Tr[\tilde{\mathbf{G}}^H\tilde{\mathbf{G}}]=MN.
%\end{align}
Besides, $\mathbb{E}[\bar{\mathbf{f}}^H\mathbf{\Phi}^H\tilde{\mathbf{G}}^H\tilde{\mathbf{G}}\mathbf{\Phi}\bar{\mathbf{f}}]$ can be computed by
\begin{align}\label{eq:Etttt}
  \mathbb{E}[\bar{\mathbf{f}}^H\mathbf{\Phi}^H\tilde{\mathbf{G}}^H\tilde{\mathbf{G}}\mathbf{\Phi}\bar{\mathbf{f}}]= \mathbb{E}\Tr[\tilde{\mathbf{G}}\mathbf{\Phi}\bar{\mathbf{f}}\bar{\mathbf{f}}^H\mathbf{\Phi}^H\tilde{\mathbf{G}}^H]
  =\mathbb{E}\Tr[\tilde{\mathbf{G}}\tilde{\mathbf{G}}^H]=MN.
\end{align}
Finally, taking the ignored coefficients into consideration, we finally obtain
\begin{equation}\label{eq:mrRiproof}
\mathbb{E}[r]=M^2N\bar{\beta}_f\bar{\beta}_G+MN(\tilde{\beta}_f\bar{\beta}_G+\bar{\beta}_f\tilde{\beta}_G+\tilde{\beta}_f\tilde{\beta}_G).
\end{equation}

The variance of $r$ can be analyzed in a similar way. Firstly, we need to compute the expectations of all the components of $r^2$. Then, we can obtain the variance of $r$ by $\mathbb{E}[(r-\mathbb{E}[r])^2]=\mathbb{E}[r^2]-\mathbb{E}[r]^2$. For ease of notation, we again analyze the expectations of the components of ${r'}^2$ at the beginning. It can be imagined that ${r'}^2$ is a summation of totally $256$ components since $r'$ has $16$ additive components. Therefore, it is quite complex to perform one-by-one analysis for the expectations of all the components. In the following, we will see that the complexity can be substantially reduced with a simple observation from \eqref{eq:r_simple_exp}. With the fact that $\mathbb{E}[\tilde{\mathbf{f}}]=\mathbf{0}$, $\mathbb{E}[\tilde{\mathbf{G}}]=\mathbf{0}$ and $\tilde{\mathbf{f}}$, $\tilde{\mathbf{G}}$ are independent of the others, the expectation of the component containing an odd number of $\tilde{\mathbf{f}}$ or $\tilde{\mathbf{G}}$ is exactly $0$. Thus, among the $256$ additive components of ${r'}^2$, we only need to focus on the square terms of the $16$ components in \eqref{eq:r_simple_exp}, and the cross product terms in which the number of $\tilde{\mathbf{f}}$ and that of $\tilde{\mathbf{G}}$ are both even. For example, $(\bar{\mathbf{f}}^H\mathbf{\Phi}^H\bar{\mathbf{G}}^H\bar{\mathbf{G}}\mathbf{\Phi}\tilde{\mathbf{f}})^{\dagger}
\bar{\mathbf{f}}^H\mathbf{\Phi}^H\tilde{\mathbf{G}}^H\tilde{\mathbf{G}}\mathbf{\Phi}\tilde{\mathbf{f}}$ contains two $\tilde{\mathbf{f}}$'s and two $\tilde{\mathbf{G}}$'s and therefore its expectation may not be $0$. Besides, its complex conjugate, i.e., $(\bar{\mathbf{f}}^H\mathbf{\Phi}^H\tilde{\mathbf{G}}^H\tilde{\mathbf{G}}\mathbf{\Phi}\tilde{\mathbf{f}})^{\dagger}
\bar{\mathbf{f}}^H\mathbf{\Phi}^H\bar{\mathbf{G}}^H\bar{\mathbf{G}}\mathbf{\Phi}\tilde{\mathbf{f}}$, which is also a component of ${r'}^2$, has a possibly non-zero expectation. Following this train of thought, the number of components with possibly non-zero expectations can be reduced to $64$, namely, $16$ square terms of the components of \eqref{eq:r_simple_exp} and $24$ pairs of components with possibly non-zero expectations. To further reduce the number of components with possibly non-zero expectation of ${r'}^2$, we need to establish the following proposition firstly.

\begin{prop}\label{prop:further_zero_comp}
Let $\mathbf{x}$ be an $N$-dimensional random vector whose entries are {\it i.i.d.} complex Gaussian random variables with zero mean and variance $2\sigma^2$, i.e., $\mathbf{x}\sim\mathcal{CN}(0, 2\sigma^2\mathbf{I}_N)$, $\mathbf{a}$, $\mathbf{b}$ be two arbitrary $N$-dimensional vectors, which are independent of $\mathbf{x}$, then we have
\begin{equation}\label{eq:2ndmoment0}
\mathbb{E}[\mathbf{x}^H\mathbf{a}\mathbf{x}^H\mathbf{b}]=0,\ {\rm or}\ \mathbb{E}[\mathbf{a}^H\mathbf{x}\mathbf{b}^H\mathbf{x}]=0.
\end{equation}
\end{prop}
\begin{proof}
Proposition \ref{prop:further_zero_comp} can be proved with the fact that $\mathbb{E}[x_m^2]=0$, where $x_m$ denotes the $m$-th entry of $\mathbf{x}$. Let $x_m=p+qi$, where $p$, $q$ denote the real part and imaginary part of $x_m$, respectively. Obviously, $p$, $q$ are two {\it i.i.d.} real Gaussian random variables with zero-mean and variance $\sigma^2$. Therefore, we have
\begin{align}\label{eq:Ex2}
\mathbb{E}[x_m^2]&=\mathbb{E}[(p+qi)^2]
=\mathbb{E}[p^2]-\mathbb{E}[q^2]+2\mathbb{E}[pq]i=0.
\end{align}
With $\mathbb{E}[x_m]=0$ and $\mathbb{E}[x_m^2]=0$, we have
\begin{align}\label{eq:axbx}
\mathbb{E}[\mathbf{a}^H\mathbf{x}\mathbf{b}^H\mathbf{x}]&=\mathbb{E}\left[\sum_{m=1}^{N}a_m^{\dagger}x_m\sum_{n=1}^{N}b_n^{\dagger}x_n\right]\nonumber\\
&=\sum_{m=1}^{N}\sum_{\substack{n=1\\n\neq m}}^{N}a_m^{\dagger}b_n^{\dagger}\mathbb{E}[x_m]\mathbb{E}[x_n]+\sum_{m=1}^{N}a_m^{\dagger}b_m^{\dagger}\mathbb{E}[x_m^2]
=0.
\end{align}
Besides, $\mathbb{E}[\mathbf{x}^H\mathbf{a}\mathbf{x}^H\mathbf{b}]=0$ can be proved in the same way.
\end{proof}

Actually, Proposition \ref{prop:further_zero_comp} provides us more useful tricks to identify the components with zero expectation of ${r'}^2$. For example, $(\bar{\mathbf{f}}^H\mathbf{\Phi}^H\bar{\mathbf{G}}^H\bar{\mathbf{G}}\mathbf{\Phi}\tilde{\mathbf{f}})^{\dagger}
\tilde{\mathbf{f}}^H\mathbf{\Phi}^H\tilde{\mathbf{G}}^H\tilde{\mathbf{G}}\mathbf{\Phi}\bar{\mathbf{f}}$ can be rewritten as $\tilde{\mathbf{f}}^H\mathbf{a}\tilde{\mathbf{f}}^H\mathbf{b}$ with $\mathbf{a}=\mathbf{\Phi}^H\bar{\mathbf{G}}^H\bar{\mathbf{G}}\mathbf{\Phi}\bar{\mathbf{f}}$ and $\mathbf{b}=\mathbf{\Phi}^H\tilde{\mathbf{G}}^H\tilde{\mathbf{G}}\mathbf{\Phi}\bar{\mathbf{f}}$. With Proposition \ref{prop:further_zero_comp}, its expectation is obviously zero. With this observation, the number of the components with non-zero expectation can be reduced to $34$, i.e., $16$ square terms of the components of \eqref{eq:r_simple_exp} and $9$ pairs of the components with possibly non-zero expectations. So far, we have to analyze the expectation of each of these components. The non-zero expectations of these components are summarized in Table \ref{tab:componentexpectations}. Specifically, $C_{1}, \cdots, C_{16}$ are exactly the $16$ square terms of the components of \eqref{eq:r_simple_exp}. It is worth noting that, for each component of $C_{17}, \cdots, C_{25}$, the expectation of itself and that of its complex conjugate are the same. Therefore, we only list one of each pair of the components with non-zero expectations in Table \ref{tab:componentexpectations}.
%Hence, we actually consider the expectation of $(C+C^{\dagger})/2$ to ensure the random variables we considered are real.
\begin{table*}[!t]
\centering
\caption{Expectations of the non-zero components of ${r'}^2$}
\label{tab:componentexpectations}
\renewcommand{\arraystretch}{1}
\resizebox{.9\textwidth }{!}{
\begin{tabular}{c|c|c|c|c|c}
  \hline
  Index & Component, $C$ & $\mathbb{E}[C]$ & Index & Component, $C$ & $\mathbb{E}[C]$ \\\hline
  $C_1$ & $(\bar{\mathbf{f}}^H\mathbf{\Phi}^H\bar{\mathbf{G}}^H\bar{\mathbf{G}}\mathbf{\Phi}\bar{\mathbf{f}})^{\dagger}
  \bar{\mathbf{f}}^H\mathbf{\Phi}^H\bar{\mathbf{G}}^H\bar{\mathbf{G}}\mathbf{\Phi}\bar{\mathbf{f}}$ & $M^4N^2$ &

  $C_2$ & $(\bar{\mathbf{f}}^H\mathbf{\Phi}^H\bar{\mathbf{G}}^H\bar{\mathbf{G}}\mathbf{\Phi}\tilde{\mathbf{f}})^{\dagger}
  \bar{\mathbf{f}}^H\mathbf{\Phi}^H\bar{\mathbf{G}}^H\bar{\mathbf{G}}\mathbf{\Phi}\tilde{\mathbf{f}}$ & $M^3N^2$ \\\hline

  $C_3$ & $(\bar{\mathbf{f}}^H\mathbf{\Phi}^H\bar{\mathbf{G}}^H\tilde{\mathbf{G}}\mathbf{\Phi}\bar{\mathbf{f}})^{\dagger}
  \bar{\mathbf{f}}^H\mathbf{\Phi}^H\bar{\mathbf{G}}^H\tilde{\mathbf{G}}\mathbf{\Phi}\bar{\mathbf{f}}$ & $M^3N$ &

  $C_4$ & $(\bar{\mathbf{f}}^H\mathbf{\Phi}^H\bar{\mathbf{G}}^H\tilde{\mathbf{G}}\mathbf{\Phi}\tilde{\mathbf{f}})^{\dagger}
  \bar{\mathbf{f}}^H\mathbf{\Phi}^H\bar{\mathbf{G}}^H\tilde{\mathbf{G}}\mathbf{\Phi}\tilde{\mathbf{f}}$ & $M^3N$ \\\hline

  $C_5$ & $(\tilde{\mathbf{f}}^H\mathbf{\Phi}^H\bar{\mathbf{G}}^H\bar{\mathbf{G}}\mathbf{\Phi}\bar{\mathbf{f}})^{\dagger}
  \tilde{\mathbf{f}}^H\mathbf{\Phi}^H\bar{\mathbf{G}}^H\bar{\mathbf{G}}\mathbf{\Phi}\bar{\mathbf{f}}$ & $M^3N^2$ &

  $C_6$ & $(\tilde{\mathbf{f}}^H\mathbf{\Phi}^H\bar{\mathbf{G}}^H\bar{\mathbf{G}}\mathbf{\Phi}\tilde{\mathbf{f}})^{\dagger}
  \tilde{\mathbf{f}}^H\mathbf{\Phi}^H\bar{\mathbf{G}}^H\bar{\mathbf{G}}\mathbf{\Phi}\tilde{\mathbf{f}}$ & $2M^2N^2$ \\\hline

  $C_7$ & $(\tilde{\mathbf{f}}^H\mathbf{\Phi}^H\bar{\mathbf{G}}^H\tilde{\mathbf{G}}\mathbf{\Phi}\bar{\mathbf{f}})^{\dagger}
  \tilde{\mathbf{f}}^H\mathbf{\Phi}^H\bar{\mathbf{G}}^H\tilde{\mathbf{G}}\mathbf{\Phi}\bar{\mathbf{f}}$ & $M^2N$ &

  $C_8$ & $(\tilde{\mathbf{f}}^H\mathbf{\Phi}^H\bar{\mathbf{G}}^H\tilde{\mathbf{G}}\mathbf{\Phi}\tilde{\mathbf{f}})^{\dagger}
  \tilde{\mathbf{f}}^H\mathbf{\Phi}^H\bar{\mathbf{G}}^H\tilde{\mathbf{G}}\mathbf{\Phi}\tilde{\mathbf{f}}$ & $M^2N$ \\\hline

  $C_9$ & $(\bar{\mathbf{f}}^H\mathbf{\Phi}^H\tilde{\mathbf{G}}^H\bar{\mathbf{G}}\mathbf{\Phi}\bar{\mathbf{f}})^{\dagger}
  \bar{\mathbf{f}}^H\mathbf{\Phi}^H\tilde{\mathbf{G}}^H\bar{\mathbf{G}}\mathbf{\Phi}\bar{\mathbf{f}}$ & $M^3N$ &

  $C_{10}$ & $(\bar{\mathbf{f}}^H\mathbf{\Phi}^H\tilde{\mathbf{G}}^H\bar{\mathbf{G}}\mathbf{\Phi}\tilde{\mathbf{f}})^{\dagger}
  \bar{\mathbf{f}}^H\mathbf{\Phi}^H\tilde{\mathbf{G}}^H\bar{\mathbf{G}}\mathbf{\Phi}\tilde{\mathbf{f}}$ & $M^2N$ \\\hline

  $C_{11}$ & $(\bar{\mathbf{f}}^H\mathbf{\Phi}^H\tilde{\mathbf{G}}^H\tilde{\mathbf{G}}\mathbf{\Phi}\bar{\mathbf{f}})^{\dagger}
  \bar{\mathbf{f}}^H\mathbf{\Phi}^H\tilde{\mathbf{G}}^H\tilde{\mathbf{G}}\mathbf{\Phi}\bar{\mathbf{f}}$ & $M^2N^2$ &

  $C_{12}$ & $(\bar{\mathbf{f}}^H\mathbf{\Phi}^H\tilde{\mathbf{G}}^H\tilde{\mathbf{G}}\mathbf{\Phi}\tilde{\mathbf{f}})^{\dagger}
  \bar{\mathbf{f}}^H\mathbf{\Phi}^H\tilde{\mathbf{G}}^H\tilde{\mathbf{G}}\mathbf{\Phi}\tilde{\mathbf{f}}$ & $(M+N)MN$ \\\hline

  $C_{13}$ & $(\tilde{\mathbf{f}}^H\mathbf{\Phi}^H\tilde{\mathbf{G}}^H\bar{\mathbf{G}}\mathbf{\Phi}\bar{\mathbf{f}})^{\dagger}
  \tilde{\mathbf{f}}^H\mathbf{\Phi}^H\tilde{\mathbf{G}}^H\bar{\mathbf{G}}\mathbf{\Phi}\bar{\mathbf{f}}$ & $M^3N$ &

  $C_{14}$ & $(\tilde{\mathbf{f}}^H\mathbf{\Phi}^H\tilde{\mathbf{G}}^H\bar{\mathbf{G}}\mathbf{\Phi}\tilde{\mathbf{f}})^{\dagger}
  \tilde{\mathbf{f}}^H\mathbf{\Phi}^H\tilde{\mathbf{G}}^H\bar{\mathbf{G}}\mathbf{\Phi}\tilde{\mathbf{f}}$ & $M^2N$ \\\hline

  $C_{15}$ & $(\tilde{\mathbf{f}}^H\mathbf{\Phi}^H\tilde{\mathbf{G}}^H\tilde{\mathbf{G}}\mathbf{\Phi}\bar{\mathbf{f}})^{\dagger}
  \tilde{\mathbf{f}}^H\mathbf{\Phi}^H\tilde{\mathbf{G}}^H\tilde{\mathbf{G}}\mathbf{\Phi}\bar{\mathbf{f}}$ & $(M+N)MN$ &

  $C_{16}$ & $(\tilde{\mathbf{f}}^H\mathbf{\Phi}^H\tilde{\mathbf{G}}^H\tilde{\mathbf{G}}\mathbf{\Phi}\tilde{\mathbf{f}})^{\dagger}
  \tilde{\mathbf{f}}^H\mathbf{\Phi}^H\tilde{\mathbf{G}}^H\tilde{\mathbf{G}}\mathbf{\Phi}\tilde{\mathbf{f}}$ & $(M+N+MN)MN$ \\\hline

  $C_{17}$ & $(\bar{\mathbf{f}}^H\mathbf{\Phi}^H\bar{\mathbf{G}}^H\bar{\mathbf{G}}\mathbf{\Phi}\bar{\mathbf{f}})^{\dagger}
  \tilde{\mathbf{f}}^H\mathbf{\Phi}^H\bar{\mathbf{G}}^H\bar{\mathbf{G}}\mathbf{\Phi}\tilde{\mathbf{f}}$ & $M^3N^2$ &

  $C_{18}$ & $(\bar{\mathbf{f}}^H\mathbf{\Phi}^H\bar{\mathbf{G}}^H\bar{\mathbf{G}}\mathbf{\Phi}\bar{\mathbf{f}})^{\dagger}
  \bar{\mathbf{f}}^H\mathbf{\Phi}^H\tilde{\mathbf{G}}^H\tilde{\mathbf{G}}\mathbf{\Phi}\bar{\mathbf{f}}$ & $M^3N^2$ \\\hline

  $C_{19}$ & $(\bar{\mathbf{f}}^H\mathbf{\Phi}^H\bar{\mathbf{G}}^H\bar{\mathbf{G}}\mathbf{\Phi}\bar{\mathbf{f}})^{\dagger}
  \tilde{\mathbf{f}}^H\mathbf{\Phi}^H\tilde{\mathbf{G}}^H\tilde{\mathbf{G}}\mathbf{\Phi}\tilde{\mathbf{f}}$ & $M^3N^2$ &

  $C_{20}$ & $(\tilde{\mathbf{f}}^H\mathbf{\Phi}^H\bar{\mathbf{G}}^H\bar{\mathbf{G}}\mathbf{\Phi}\tilde{\mathbf{f}})^{\dagger}
  \bar{\mathbf{f}}^H\mathbf{\Phi}^H\tilde{\mathbf{G}}^H\tilde{\mathbf{G}}\mathbf{\Phi}\bar{\mathbf{f}}$ & $M^2N^2$ \\\hline

  $C_{21}$ & $(\tilde{\mathbf{f}}^H\mathbf{\Phi}^H\bar{\mathbf{G}}^H\bar{\mathbf{G}}\mathbf{\Phi}\tilde{\mathbf{f}})^{\dagger}
  \tilde{\mathbf{f}}^H\mathbf{\Phi}^H\tilde{\mathbf{G}}^H\tilde{\mathbf{G}}\mathbf{\Phi}\tilde{\mathbf{f}}$ & $M^2N^2$ &

  $C_{22}$ & $(\bar{\mathbf{f}}^H\mathbf{\Phi}^H\tilde{\mathbf{G}}^H\tilde{\mathbf{G}}\mathbf{\Phi}\bar{\mathbf{f}})^{\dagger}
  \tilde{\mathbf{f}}^H\mathbf{\Phi}^H\tilde{\mathbf{G}}^H\tilde{\mathbf{G}}\mathbf{\Phi}\tilde{\mathbf{f}}$ & $M^2N^2$ \\\hline

  $C_{23}$ & $(\bar{\mathbf{f}}^H\mathbf{\Phi}^H\bar{\mathbf{G}}^H\bar{\mathbf{G}}\mathbf{\Phi}\tilde{\mathbf{f}})^{\dagger}
  \bar{\mathbf{f}}^H\mathbf{\Phi}^H\tilde{\mathbf{G}}^H\tilde{\mathbf{G}}\mathbf{\Phi}\tilde{\mathbf{f}}$ & $M^2N^2$ &

  $C_{24}$ & $(\tilde{\mathbf{f}}^H\mathbf{\Phi}^H\bar{\mathbf{G}}^H\bar{\mathbf{G}}\mathbf{\Phi}\bar{\mathbf{f}})^{\dagger}
  \tilde{\mathbf{f}}^H\mathbf{\Phi}^H\tilde{\mathbf{G}}^H\tilde{\mathbf{G}}\mathbf{\Phi}\bar{\mathbf{f}}$ & $M^2N^2$ \\\hline

  $C_{25}$ & $(\bar{\mathbf{f}}^H\mathbf{\Phi}^H\tilde{\mathbf{G}}^H\bar{\mathbf{G}}\mathbf{\Phi}\bar{\mathbf{f}})^{\dagger}
  \tilde{\mathbf{f}}^H\mathbf{\Phi}^H\tilde{\mathbf{G}}^H\bar{\mathbf{G}}\mathbf{\Phi}\tilde{\mathbf{f}}$ & $M^2N$ &

  --&--&--\\\hline

\end{tabular}
}
\end{table*}

Most results in Table \ref{tab:componentexpectations} can be directly derived with Lemma \ref{lem:rdfRay} and the cyclic property of the trace of a product, i.e.,
\begin{equation}\label{eq:cyclicproperty}
\Tr[\mathbf{ABC}]=\Tr[\mathbf{BCA}]=\Tr[\mathbf{CAB}],
\end{equation}
where $\mathbf{A}$, $\mathbf{B}$, $\mathbf{C}$ are three arbitrary matrices.
In particular, to derive the expectations of $C_{6}$ and $C_{16}$, we also requires the following formula:
\begin{equation}\label{eq:varm1m2}
\mathbb{E}[(x-\mathbb{E}[x])^2]=\mathbb{E}[x^2]-\mathbb{E}[x]^2,
\end{equation}
where $x$ is a real random variable with finite second order moment. We take $C_{16}$ as a example to show the details. With Lemma \ref{lem:rdfRay}, $\tilde{\mathbf{f}}^H\mathbf{\Phi}^H\tilde{\mathbf{G}}^H\tilde{\mathbf{G}}\mathbf{\Phi}\tilde{\mathbf{f}}$ converges to a real Gaussian distribution of mean $MN$ and variance $(M+N)MN$. Then, with \eqref{eq:varm1m2}, we have
\begin{align}\label{eq:exp16}
&\mathbb{E}[(\tilde{\mathbf{f}}^H\mathbf{\Phi}^H\tilde{\mathbf{G}}^H\tilde{\mathbf{G}}\mathbf{\Phi}\tilde{\mathbf{f}})^{\dagger}
\tilde{\mathbf{f}}^H\mathbf{\Phi}^H\tilde{\mathbf{G}}^H\tilde{\mathbf{G}}\mathbf{\Phi}\tilde{\mathbf{f}}]
=\mathbb{E}[(\tilde{\mathbf{f}}^H\mathbf{\Phi}^H\tilde{\mathbf{G}}^H\tilde{\mathbf{G}}\mathbf{\Phi}\tilde{\mathbf{f}})^{2}]\nonumber\\
&=(MN)^2+(M+N)MN
=(M+N+MN)MN.
\end{align}

%The expectation of other components can be derived in a similar way and the corresponding results are listed in Table \ref{tab:componentexpectations}.

Finally, we take the ignored coefficients, i.e., $\bar{\beta}_f$, $\tilde{\beta}_f$, $\bar{\beta}_G$, $\tilde{\beta}_G$, into consideration to obtain $\mathbb{E}[r^2]$. Denoting $\beta_i$ the coefficient of $C_{i}(i=1, \cdots, 25)$ in $r^2$, $\mathbb{E}[r^2]$ can be written as
\begin{equation}\label{eq:Er2sum}
\mathbb{E}[r^2]=\sum_{i=1}^{16}\beta_i\mathbb{E}[C_{i}]+2\sum_{i=17}^{25}\beta_i\mathbb{E}[C_{i}],
\end{equation}
and the result is exactly
\begin{align}\label{eq:Er2}
\mathbb{E}[r^2]&=M^4N^2\bar{\beta}_f\bar{\beta}_G+
4M^3N^2\bar{\beta}_f\tilde{\beta}_f\bar{\beta}_G+
(2M^3N+2M^3N^2)\bar{\beta}_f\bar{\beta}_G\tilde{\beta}_G+
(M^2N^2)\bar{\beta}_f^2\tilde{\beta}_G^2\nonumber\\
&\quad+(M+N+MN)MN\tilde{\beta}_f^2\tilde{\beta}_G^2
+(2M^3N^2+6M^2N^2+2M^3N+4M^2N)\bar{\beta}_f\tilde{\beta}_f\bar{\beta}_G\tilde{\beta}_G\nonumber\\
&\quad+2M^2N^2\tilde{\beta}_f^2\bar{\beta}_G^2+
(2M^2N^2+2M^2N)\tilde{\beta}_f^2\bar{\beta}_G\tilde{\beta}_G.
\end{align}
The variance of $r$ can then be calculated by
\begin{align}\label{eq:vrRiproof}
v_r^{\rm Ri} = \mathbb{E}[(r-\mathbb{E}[r])^2]=\mathbb{E}[r^2]-\mathbb{E}[r]^2,
\end{align}
where $\mathbb{E}[r^2]$, $\mathbb{E}[r]$ are defined in \eqref{eq:mrRiproof} and \eqref{eq:Er2}, respectively.
\bibliographystyle{IEEEtran}
\bibliography{IEEEabrv, ref_RIS4SS}

\begin{thebibliography}{10}
\providecommand{\url}[1]{#1}
\csname url@samestyle\endcsname
\providecommand{\newblock}{\relax}
\providecommand{\bibinfo}[2]{#2}
\providecommand{\BIBentrySTDinterwordspacing}{\spaceskip=0pt\relax}
\providecommand{\BIBentryALTinterwordstretchfactor}{4}
\providecommand{\BIBentryALTinterwordspacing}{\spaceskip=\fontdimen2\font plus
\BIBentryALTinterwordstretchfactor\fontdimen3\font minus
  \fontdimen4\font\relax}
\providecommand{\BIBforeignlanguage}[2]{{%
\expandafter\ifx\csname l@#1\endcsname\relax
\typeout{** WARNING: IEEEtran.bst: No hyphenation pattern has been}%
\typeout{** loaded for the language `#1'. Using the pattern for}%
\typeout{** the default language instead.}%
\else
\language=\csname l@#1\endcsname
\fi
#2}}
\providecommand{\BIBdecl}{\relax}
\BIBdecl

\bibitem{liang2020dynamic}
Y.-C. Liang, \emph{Dynamic spectrum management: from cognitive radio to
  blockchain and artificial intelligence}.\hskip 1em plus 0.5em minus
  0.4em\relax Springer, 2020.

\bibitem{liang2008sensing}
Y.-C. Liang, Y.~Zeng, E.~C. Peh, and A.~T. Hoang, ``Sensing-throughput tradeoff
  for cognitive radio networks,'' \emph{{IEEE} Trans. Wireless Commun.},
  vol.~7, no.~4, pp. 1326--1337, Apr. 2008.

\bibitem{sonnenschein1992radiometric}
A.~Sonnenschein and P.~M. Fishman, ``Radiometric detection of spread-spectrum
  signals in noise of uncertain power,'' \emph{{IEEE} Trans. Aerosp. Electron.
  Syst.}, vol.~28, no.~3, pp. 654--660, 1992.

\bibitem{sahai2005maximum}
A.~Sahai and D.~Cabric, ``Spectrum sensing: Fundamental limits and practical
  challenges,'' in \emph{Proc. IEEE Int. Symp. New Frontiers in Dynamic
  Spectrum Access Networks (DySPAN)}, Baltimore, BD, Nov. 2005.

\bibitem{chen2007signature}
H.-S. Chen, W.~Gao, and D.~G. Daut, ``Signature based spectrum sensing
  algorithms for {IEEE 802.22 WRAN},'' in \emph{Proc. {IEEE} Int. Conf. Commun.
  ({ICC})}, Glasgow, UK, 2007, pp. 6487--6492.

\bibitem{gardner1991exploitation}
W.~A. Gardner, ``Exploitation of spectral redundancy in cyclostationary
  signals,'' \emph{IEEE Signal Process. Mag.}, vol.~8, no.~2, pp. 14--36, Apr.
  1991.

\bibitem{han2006spectral}
N.~Han, S.~Shon, J.~H. Chung, and J.~M. Kim, ``Spectral correlation based
  signal detection method for spectrum sensing in {IEEE} 802.22 {WRAN}
  systems,'' in \emph{Proc. Int. Conf. Advanced Communication Technology},
  Phoenix Park, Korea, 2006, pp. 1765--1770.

\bibitem{zeng2008maximum}
Y.~Zeng, C.~L. Koh, and Y.-C. Liang, ``Maximum eigenvalue detection: Theory and
  application,'' in \emph{Proc. {IEEE} Int. Conf. Commun. ({ICC})}, Beijing,
  China, 2008, pp. 4160--4164.

\bibitem{zeng2009eigenvalue}
Y.~Zeng and Y.-C. Liang, ``Eigenvalue-based spectrum sensing algorithms for
  cognitive radio,'' \emph{{IEEE} Trans. Commun.}, vol.~57, no.~6, pp.
  1784--1793, Jun. 2009.

\bibitem{zhang2010multi}
R.~Zhang, T.~J. Lim, Y.-C. Liang, and Y.~Zeng, ``Multi-antenna based spectrum
  sensing for cognitive radios: A {GLRT} approach,'' \emph{{IEEE} Trans.
  Commun.}, vol.~58, no.~1, pp. 84--88, Jan. 2010.

\bibitem{bouallegue2017blind}
K.~Bouallegue, I.~Dayoub, M.~Gharbi, and K.~Hassan, ``Blind spectrum sensing
  using extreme eigenvalues for cognitive radio networks,'' \emph{{IEEE}
  Commun. Lett.}, vol.~22, no.~7, pp. 1386--1389, Jul. 2017.

\bibitem{yucek2009survey}
T.~Yucek and H.~Arslan, ``A survey of spectrum sensing algorithms for cognitive
  radio applications,'' \emph{{IEEE} Commun. Surveys Tuts.}, vol.~11, no.~1,
  pp. 116--130, Firstquarter 2009.

\bibitem{zeng2010review}
Y.~Zeng, Y.-C. Liang, A.~T. Hoang, and R.~Zhang, ``A review on spectrum sensing
  for cognitive radio: challenges and solutions,'' \emph{EURASIP J. Adv. Signal
  Process.}, vol. 2010, pp. 1--15, 2010.

\bibitem{awin2019blind}
F.~Awin, E.~Abdel-Raheem, and K.~Tepe, ``Blind spectrum sensing approaches for
  interweaved cognitive radio system: A tutorial and short course,''
  \emph{{IEEE} Commun. Surveys Tuts.}, vol.~21, no.~1, pp. 238--259,
  Firstquarter 2019.

\bibitem{liu2019deep}
C.~Liu, J.~Wang, X.~Liu, and Y.-C. Liang, ``Deep {CM-CNN} for spectrum sensing
  in cognitive radio,'' \emph{{IEEE} J. Sel. Areas Commun.}, vol.~37, no.~10,
  pp. 2306--2321, Oct. 2019.

\bibitem{gao2019deep}
J.~Gao, X.~Yi, C.~Zhong, X.~Chen, and Z.~Zhang, ``Deep learning for spectrum
  sensing,'' \emph{{IEEE} Wireless Commun. Lett.}, vol.~8, no.~6, pp.
  1727--1730, Dec. 2019.

\bibitem{bianchi2011performance}
P.~Bianchi, M.~Debbah, M.~Ma{\"\i}da, and J.~Najim, ``Performance of
  statistical tests for single-source detection using random matrix theory,''
  \emph{{IEEE} Trans. Inf. Theory}, vol.~57, no.~4, pp. 2400--2419, Mar. 2011.

\bibitem{jin2012performance}
M.~Jin, Y.~Li, and H.-G. Ryu, ``On the performance of covariance based spectrum
  sensing for cognitive radio,'' \emph{{IEEE} Trans. Signal Process.}, vol.~60,
  no.~7, pp. 3670--3682, Jul. 2012.

\bibitem{wei2012spectrum}
L.~Wei and O.~Tirkkonen, ``Spectrum sensing in the presence of multiple primary
  users,'' \emph{{IEEE} Trans. Commun.}, vol.~60, no.~5, pp. 1268--1277, May
  2012.

\bibitem{wei2014multi}
L.~Wei, O.~Tirkkonen, and Y.-C. Liang, ``Multi-source signal detection with
  arbitrary noise covariance,'' \emph{{IEEE} Trans. Signal Process.}, vol.~62,
  no.~22, pp. 5907--5918, Nov. 2014.

\bibitem{sedighi2015performance}
S.~Sedighi, A.~Taherpour, J.~Sala-Alvarez, and T.~Khattab, ``On the performance
  of {Hadamard} ratio detector-based spectrum sensing for cognitive radios,''
  \emph{{IEEE} Trans. Signal Process.}, vol.~63, no.~14, pp. 3809--3824, Jul.
  2015.

\bibitem{sedighi2016eigenvalue}
S.~Sedighi, A.~Taherpour, S.~Gazor, and T.~Khattab, ``Eigenvalue-based multiple
  antenna spectrum sensing: Higher order moments,'' \emph{{IEEE} Trans.
  Wireless Commun.}, vol.~16, no.~2, pp. 1168--1184, Feb. 2016.

\bibitem{liang2019large}
Y.-C. Liang, R.~Long, Q.~Zhang, J.~Chen, H.~V. Cheng, and H.~Guo, ``Large
  intelligent surface/antennas ({LISA}): Making reflective radios smart,''
  \emph{J. Commun. Inf. Netw.}, vol.~4, no.~2, pp. 40--50, Jun. 2019.

\bibitem{wu2019towards}
Q.~Wu and R.~Zhang, ``Towards smart and reconfigurable environment: Intelligent
  reflecting surface aided wireless network,'' \emph{{IEEE} Commun. Mag.},
  vol.~58, no.~1, pp. 106--112, Jan. 2019.

\bibitem{gong2020toward}
S.~Gong, X.~Lu, D.~T. Hoang, D.~Niyato, L.~Shu, D.~I. Kim, and Y.-C. Liang,
  ``Toward smart wireless communications via intelligent reflecting surfaces: A
  contemporary survey,'' \emph{{IEEE} Commun. Surveys Tuts.}, vol.~22, no.~4,
  pp. 2283--2314, Fourthquarter 2020.

\bibitem{cui2014coding}
T.~J. Cui, M.~Q. Qi, X.~Wan, J.~Zhao, and Q.~Cheng, ``Coding metamaterials,
  digital metamaterials and programmable metamaterials,'' \emph{Light: Science
  \& Applications}, vol.~3, no.~10, p. e218, 2014.

\bibitem{chen2019intelligent}
J.~Chen, Y.-C. Liang, Y.~Pei, and H.~Guo, ``Intelligent reflecting surface: A
  programmable wireless environment for physical layer security,'' \emph{IEEE
  Access}, vol.~7, pp. 82\,599--82\,612, 2019.

\bibitem{guo2020weighted}
H.~Guo, Y.-C. Liang, J.~Chen, and E.~G. Larsson, ``Weighted sum-rate
  maximization for reconfigurable intelligent surface aided wireless
  networks,'' \emph{{IEEE} Trans. Wireless Commun.}, vol.~19, no.~5, pp.
  3064--3076, May 2020.

\bibitem{huang2019reconfigurable}
C.~Huang, A.~Zappone, G.~C. Alexandropoulos, M.~Debbah, and C.~Yuen,
  ``Reconfigurable intelligent surfaces for energy efficiency in wireless
  communication,'' \emph{{IEEE} Trans. Wireless Commun.}, vol.~18, no.~8, pp.
  4157--4170, Aug. 2019.

\bibitem{yuan2020intelligent}
J.~Yuan, Y.-C. Liang, J.~Joung, G.~Feng, and E.~G. Larsson, ``Intelligent
  reflecting surface-assisted cognitive radio system,'' \emph{{IEEE} Trans.
  Commun.}, vol.~69, no.~1, pp. 675--687, Jan. 2021.

\bibitem{chen2019channel}
J.~Chen, Y.-C. Liang, H.~V. Cheng, and W.~Yu, ``Channel estimation for
  reconfigurable intelligent surface aided multi-user {MIMO} systems,''
  \emph{arXiv preprint arXiv:1912.03619}, 2019.

\bibitem{zhou2020framework}
G.~Zhou, C.~Pan, H.~Ren, K.~Wang, and A.~Nallanathan, ``A framework of robust
  transmission design for {IRS}-aided {MISO} communications with imperfect
  cascaded channels,'' \emph{{IEEE} Trans. Signal Process.}, vol.~68, pp.
  5092--5106, Aug. 2020.

\bibitem{zhou2020robust}
G.~Zhou, C.~Pan, H.~Ren, K.~Wang, M.~Di~Renzo, and A.~Nallanathan, ``Robust
  beamforming design for intelligent reflecting surface aided {MISO}
  communication systems,'' \emph{{IEEE} Wireless Commun. Lett.}, vol.~9,
  no.~10, pp. 1658--1662, Oct. 2020.

\bibitem{kammoun2020asymptotic}
A.~Kammoun, A.~Chaaban, M.~Debbah, and M.-S. Alouini, ``Asymptotic max-min
  {SINR} analysis of reconfigurable intelligent surface assisted {MISO}
  systems,'' \emph{{IEEE} Trans. Wireless Commun.}, vol.~19, no.~12, pp.
  7748--7764, Dec. 2020.

\bibitem{zhang2020transmitter}
J.~Zhang, J.~Liu, S.~Ma, C.-K. Wen, and S.~Jin, ``Transmitter design for large
  intelligent surface-assisted {MIMO} wireless communication with statistical
  {CSI},'' in \emph{Proc. IEEE Int. Conf. on Commun. Workshop}, Dublin,
  Ireland, 2020, pp. 1--5.

\bibitem{zhang2021large}
J.~Zhang, J.~Liu, S.~Ma, C.-K. Wen, and S.~Jin, ``Large system achievable rate
  analysis of {RIS}-assisted {MIMO} wireless communication with statistical
  {CSIT},'' \emph{IEEE Trans. Wireless Commun., Early Access}, 2021,
  doi:10.1109/TWC.2021.3068494.

\bibitem{li2020irs}
X.~Li, Q.~Zhu, and Y.~Wang, ``{IRS}-assisted crowd spectrum sensing in {B5G}
  cellular {IoT} networks,'' in \emph{Proc. {IEEE} Int. Conf. Wireless Commun.
  Signal Process. ({WCSP})}, Nanjing, China, 2020, pp. 761--765.

\bibitem{ge2021large}
J.~Ge, Y.-C. Liang, Z.~Bai, and G.~Pan, ``Large-dimensional random matrix
  theory and its applications in deep learning and wireless communications,''
  \emph{Random Matrices: Theory Appl.}, p. 2230001, 2021.

\bibitem{RMTstat}
I.~M. Johnstone, Z.~Ma, P.~O. Perry, and M.~Shahram, \emph{RMTstat:
  Distributions, Statistics and Tests derived from Random Matrix Theory}, 2014,
  r package version 0.3.

\bibitem{griffin2009complete}
J.~D. Griffin and G.~D. Durgin, ``Complete link budgets for backscatter-radio
  and {RFID} systems,'' \emph{{IEEE} Antennas Propag. Mag.}, vol.~51, no.~2,
  pp. 11--25, Jul. 2009.

\bibitem{couillet2011random}
R.~Couillet and M.~Debbah, \emph{Random matrix methods for wireless
  communications}.\hskip 1em plus 0.5em minus 0.4em\relax Cambridge University
  Press, 2011.

\bibitem{tse2000linear}
D.~N.~C. Tse and O.~Zeitouni, ``Linear multiuser receivers in random
  environments,'' \emph{{IEEE} Trans. Inf. Theory}, vol.~46, no.~1, pp.
  171--188, Jan. 2000.

\bibitem{liang2007asymptotic}
Y.-C. Liang, G.~Pan, and Z.~Bai, ``Asymptotic performance of {MMSE} receivers
  for large systems using random matrix theory,'' \emph{{IEEE} Trans. Inf.
  Theory}, vol.~53, no.~11, pp. 4173--4190, Oct. 2007.

\bibitem{fang2014spectral}
Z.~Fang, Y.-C. Liang, and Z.~Bai, \emph{Spectral theory of large dimensional
  random matrices and its applications to wireless communications and finance
  statistics: random matrix theory and its applications}.\hskip 1em plus 0.5em
  minus 0.4em\relax World Scientific, 2014.

\bibitem{bai2007asymptotics}
Z.~D. Bai, B.~Q. Miao, and G.~M. Pan, ``On asymptotics of eigenvectors of large
  sample covariance matrix,'' \emph{Ann. Probab.}, vol.~35, no.~4, pp. 1532 --
  1572, 2007.

\end{thebibliography}

\end{document}